\documentclass[a4paper, fleqn]{llncs}

\usepackage{latexsym}
\usepackage{amsmath}
\usepackage{amssymb}
\usepackage[utf8]{inputenc}
\usepackage{tikz}
\usetikzlibrary{automata}
\usepackage{stmaryrd}

\newcommand{\quot}[1]{``#1''}
\newcommand{\myquot}[1]{``#1''}

\newcommand{\coloneq}{\mathop{:=}}

\newcommand{\cceq}{\mathop{::=}}

\renewcommand{\epsilon}{\varepsilon}

\newcommand{\pow}[1]{2^{#1}}
\newcommand{\nats}{\mathbb{N}}
\newcommand{\card}[1]{|#1|}
\newcommand{\size}[1]{\card{#1}}
\newcommand{\set}[1]{\{#1\}}

\newcommand{\aut}{\mathfrak{A}}

\newcommand{\arena}{\mathcal{A}}
\newcommand{\game}{\mathcal{G}}
\newcommand{\col}{\Omega}

\newcommand{\halfthinspace}{{\kern .08333em}}

\newcommand{\ttrue}{\texttt{tt}}
\newcommand{\ffalse}{\texttt{ff}}
\newcommand{\F}{\mathop{\mathbf{F}}}

\newcommand{\U}{\mathbin{\mathbf{U}}}

\newcommand{\conc}{\,;}
\newcommand{\Var}{\mathcal{V}}
\newcommand{\cl}{\mathrm{cl}}
\newcommand{\var}{\mathrm{var}}
\newcommand{\vardiamond}{\mathrm{var}_{\Diamond}}
\newcommand{\varbox}{\mathrm{var}_{\Box}}

\newcommand{\ddiamond}[1]{\langle\/ #1 \/\rangle\,}
\newcommand{\bbox}[1]{[\halfthinspace#1\halfthinspace]\,}

\newcommand{\ddiamondle}[2]{{\langle\/ #1 \/\rangle}_{\!\le #2}\,}
\newcommand{\bboxle}[2]{{[\halfthinspace#1\halfthinspace]}_{\le #2}\,}

\newcommand{\ddiamondcp}[1]{{\langle\/ #1 \/\rangle}_{\!\mathit{cp}}\,}
\newcommand{\bboxcp}[1]{{[\halfthinspace#1\halfthinspace]}_{\mathit{cp}}\,}

\newcommand{\Rexp}{\mathcal{R}}
\newcommand{\rel}[1]{\mathrm{rel}(#1)}

\newcommand{\dual}[1]{\overline{#1}}

\newcommand{\ltl}{\text{LTL}}

\newcommand{\pltl}{\text{PLTL}}
\newcommand{\prompt}{\text{PROMPT}$\textendash$\ltl}

\newcommand{\pldl}{\text{PLDL}}
\newcommand{\ldl}{\text{LDL}}
\newcommand{\ldlt}{\text{LDL}_{cp}}
\newcommand{\pldldiamond}{\text{PLDL}_{\Diamond}}
\newcommand{\pldlbox}{\text{PLDL}_{\Box}}

\newcommand{\nlogspace}{{\textsc{NLogSpace}}}

\newcommand{\pspace}{\textsc{{PSpace}}}
\newcommand{\twoexp}{\textsc{{2ExpTime}}}

\newcommand{\bplus}{\mathcal{B}^{+}}
\newcommand{\trace}{\mathrm{tr}}
\newcommand{\suc}[2]{\mathrm{Succ}_{#1}{#2}}
\newcommand{\outcome}{\mathrm{outcome}}
\newcommand{\trans}{\mathcal{T}}
\newcommand{\sys}{\mathcal{S}}
\newcommand{\marking}{m}

\newcommand{\update}[2]{\mathrm{upd}(#1,#2)}

\newcommand{\phiass}{\varphi_{\!A}}
\newcommand{\phigua}{\varphi_G}
\newcommand{\agmc}[3]{\langle #1 \rangle #2 \langle #3 \rangle}
\newcommand{\agmcdflt}{\agmc{\phiass}{\sys}{\phigua}}
\newcommand{\parcomp}{\!\parallel\!}

\title{Parametric Linear Dynamic Logic (full version)\thanks{A preliminary version of this work appeared in GandALF 2014. The research leading to this work was partially supported by the projects ``TriCS'' (ZI 1516/1-1) and ``AVACS'' (SFB/TR 14) of the German Research Foundation (DFG).}}

\author{Peter Faymonville and Martin Zimmermann}

\institute{Reactive Systems Group, Saarland University, 66123 Saarbrücken, Germany\\
 \email{\{faymonville, zimmermann\}@react.uni-saarland.de}}

\begin{document}

\maketitle

\begin{abstract}
We introduce Parametric Linear Dynamic Logic (PLDL), which extends Linear Dynamic Logic (LDL) by adding temporal operators equipped with parameters that bound their scope. LDL itself was proposed as an extension of Linear Temporal Logic (LTL) that is able to express all $\omega$-regular specifications while still maintaining many of LTL's desirable properties like intuitive syntax and semantics and a translation into non-deterministic Büchi automata of exponential size. However, LDL lacks capabilities to express timing constraints. By adding parameterized operators to LDL, we obtain a logic which is able to express all $\omega$-regular properties and which subsumes parameterized extensions of LTL like Parametric LTL and PROMPT-LTL.

Our main technical contribution is a translation of PLDL formulas into non-deterministic Büchi  automata of exponential size via alternating automata. This yields polynomial space algorithms for model checking and assume-guarantee model checking and a realizability algorithm with doubly-exponential running time. All three problems are also shown to be complete for these complexity classes. Moreover, we give tight upper and lower bounds on optimal parameter values for model checking and realizability. Using these bounds, we present a polynomial space procedure for model checking optimization and an algorithm with triply-exponential running time for realizability optimization. Our results show that PLDL model checking, assume-guarantee model checking, and realizability are no harder than their respective (parametric) LTL counterparts.
\end{abstract}

\section{Introduction}
\label{sec_intro}
Linear Temporal Logic ($\ltl$) \cite{Pnueli77} is a popular specification language for the verification and synthesis of reactive systems and provides semantic foundations for industrial logics like PSL~\cite{EisnerFismanPSL}. $\ltl$ has a number of desirable properties contributing to its ongoing popularity: it does not rely on the use of variables, it has an intuitive syntax and semantics and thus gives a way for practitioners to write declarative and concise specifications. Furthermore, it is expressively equivalent to first-order logic over the natural numbers with successor and order~\cite{Kamp68} and enjoys an exponential compilation property: one can efficiently construct a language-equivalent non-deterministic Büchi automaton of exponential size in the size of the specification.	
The exponential compilation property yields a polynomial space model checking algorithm and a doubly-exponential time algorithm for realizability. Both problems are complete for the respective classes.
	
	Model checking of properties described in $\ltl$ or its practical descendants is routinely applied in industrial-sized applications, especially for hardware systems \cite{EisnerFismanPSL,Forspec02}. Due to its complexity, realizability has not reached industrial acceptance (yet). First approaches relied on determinization of $\omega$-automata, which is notoriously hard to implement efficiently \cite{AlthoffThomasWallmeier2006}. More recent algorithms for realizability follow a safraless construction \cite{FiliotJinRaskin2011,FinkbeinerSchewe2013}, which avoids explicitly constructing the deterministic automaton, and show promise on small examples. 

Despite the desirable properties, two drawbacks of $\ltl$ remain and are tackled by different approaches in the literature: first, $\ltl$ is not able to express all $\omega$-regular properties. For example, the property ``$p$ holds on every even step'' (but may or may not hold on odd steps) is not expressible in $\ltl$, but is easily expressible by an $\omega$-regular expression. This drawback is a serious one, since the combination of regular properties and linear-time operators is common in hardware verification languages to express modular verification properties, as in ForSpec~\cite{Forspec02}.
 Several extensions of $\ltl$ with regular expressions, finite automata, or grammar operators~\cite{LeuckerSanchez07, VardiWolper94, Wolper1983} have been proposed as a remedy.  

A second drawback of classic temporal logics like $\ltl$ is the inability to natively express timing constraints. The standard semantics are unable to enforce the fulfillment of eventualities within finite time bounds, e.g., it is impossible to require that requests are granted within a fixed, but arbitrary, amount of time. While it is possible to unroll an a-priori fixed bound for an eventuality into $\ltl$, this requires prior knowledge of the system's granularity and incurs a blow-up when translated to automata, and is thus considered impractical.
 A more practical way of fixing this drawback is the purpose of a long line of work in parametric temporal logics, e.g., parametric $\ltl$~($\pltl$)~\cite{AlurEtessamiLaTorrePeled01}, $\prompt$~\cite{KupfermanPitermanVardi09} and parametric MITL~\cite{GiampaoloTorreNapoli15}. These logics feature parameterized temporal operators to express time bounds, and either test the existence of a global bound, like $\prompt$, or of individual bounds on the parameters, like $\pltl$.  

Recently, the first drawback was revisited by De Giacomo and Vardi ~\cite{GiacomoVardi13, Vardi11} by introducing an extension of $\ltl$ called linear dynamic logic ($\ldl$), which is as expressive as $\omega$-regular languages. The syntax of $\ldl$ is inspired by propositional dynamic logic (PDL)~\cite{FischerLadner1979}, but the semantics follow linear-time logics. In PDL and $\ldl$, systems are expressed by regular expressions~$r$ with tests, and temporal requirements are specified by two basic modalities: 
\begin{itemize}
\item	 $\ddiamond{r}\varphi$, stating that $\varphi$ should hold at some position where $r$ matches, and
\item $\bbox{r}\varphi$, stating that $\varphi$ should hold at all positions where $r$ matches.
\end{itemize}

The operators to build regular expressions from propositional formulas are as follows: sequential composition ($r_1 \conc r_2$), non-deterministic choice ($r_1 + r_2$), repetition ($r^*$), and test $(\varphi?)$ of a temporal formula. On the level of the temporal operators, conjunction and disjunction are allowed. The tests allow to check temporal properties within regular expressions, and are used to encode $\ltl$ into $\ldl$. 

For example, the program ``{\bf while} $q$ {\bf do} $a$'' with property $p$ holding after termination of the loop is expressed in PDL/$\ldl$ as follows: 
\begin{equation*}
	\bbox{(q?\conc a)^*\conc \neg q? } p \, .
\end{equation*}
Intuitively, the loop condition $q$ is tested on every loop entry, the loop body $a$ is executed/consumed until $\neg q$ holds, and then the post-condition $p$ has to hold. 

A request-response property (every request should eventually be responded to) can be formalized as follows:
\begin{equation*}
	 \bbox{\ttrue^*} (\mathit{req} \rightarrow \ddiamond{\ttrue^*} \mathit{resp})\, .
\end{equation*}

Both aforementioned drawbacks of $\ltl$, the inability to express all $\omega$-regular properties and the missing capability to specify timing constraints, have been tackled individually in a successful way in previous work, but not at the same time. Here, we propose a logic called $\pldl$ that combines the expressivity of $\ldl$ with the parametricity of P$\ltl$.

	In $\pldl$, we are for example able to parameterize the eventuality of the  request-response condition, denoted as 
	\begin{equation*}
			 {\bbox{\ttrue^*}(\mathit{req}\rightarrow\ddiamondle{\ttrue^*}{ x}\mathit{resp})} \, , 
	\end{equation*}
	 which states that every request has to be followed by a response within $x$ steps. 
	 
	Finally, the aforementioned property that is not expressible in $\ltl$, (``p holds on every even step'') can be expressed in $\pldl$ as
	\begin{equation*}
		\bbox{(\ttrue \conc \ttrue)^*} p \, .
	\end{equation*}
	 
Using the parameterized request-response property as the specification for a model checking problem entails determining whether there exists a valuation $\alpha(x)$ for $x$ such that all paths of a given system respond to requests within $\alpha(x)$ steps.
	
If we take the property as a specification for the $\pldl$ realizability problem, and define $\mathit{req}$ as input, $\mathit{resp}$ as output, we compute whether there exists a winning strategy that adheres to a valuation $\alpha(x)$ and therefore ensures the delivery of responses to requests in a timely manner.

The main result of this paper is the translation of $\pldl$ into alternating Büchi automata of linear size. Using these automata and a  generalization of the alternating color technique of \cite{KupfermanPitermanVardi09}, we obtain the following results.

First, we prove that $\pldl$ model checking is $\pspace$-complete by constructing a non-deterministic Büchi automaton of exponential size and using a modified on-the-fly non-emptiness test to obtain membership in $\pspace$. $\pspace$-hardness follows from the conversion of $\ltl$ to $\pldl$. Furthermore, we give a tight exponential bound on the satisfying valuation for model checking.

Second, we consider the $\pldl$ assume-guarantee model checking problem and show it to be $\pspace$-complete as well by extending the techniques used to show the similar result for model checking.

Third, we prove that $\pldl$ realizability is $\twoexp$-complete. Hardness again follows from the ability to express $\ltl$, while membership is proven by solving a parity game constructed from a deterministic parity automaton of doubly-exponential size. Additionally, we give a tight doubly-exponential bound on the satisfying valuation for realizability.

Thus, the model checking, the assume-guarantee model checking, and the realizability problem are no harder than their corresponding variants for $\ltl$. All three solutions to these problems are extensions of the ones for $\prompt$~\cite{KupfermanPitermanVardi09}.

Fourth, we investigate optimization problems for $\pldl$, i.e., determining the optimal valuation for a formula and a system or the tightest guarantee for realizability. While the model checking optimization problem is still solvable in polynomial space, we provide a triply-exponential time algorithm for the realizability optimization problem. This leaves an exponential gap to the decision variant, as for $\pltl$~\cite{Zimmermann13}. Both algorithms are based on exhaustive search through the bounded solution space induced by the upper bounds mentioned above. 

Our translation into alternating automata is also applicable to $\ldl$ on infinite traces, while De~Giacomo and Vardi~\cite{GiacomoVardi13} only considered $\ldl$ on finite traces.  Furthermore, our construction differs conceptually, since we present a bottom-up procedure, while they gave a top-down construction.

\section{PLDL}
\label{sec_defs}
Let $\Var$ be a set of variables and let us fix a finite set~$P$ of atomic propositions which we use to build formulas and to label transition systems. For a subset $A \in \pow{P}$ and a propositional formula~$\phi$ over $P$, we write $A \models \phi$, if the variable valuation mapping elements in $A$ to true and elements not in $A$ to false satisfies $\phi$. The formulas of $\pldl$ are given by the grammar
\begin{align*}
\varphi &\cceq p \mid \neg p \mid \varphi \wedge \varphi \mid \varphi \vee \varphi
  \mid \ddiamond{r} \varphi 
  \mid \bbox{r} \varphi 
  \mid \ddiamondle{r}{z} \varphi 
  \mid \bboxle{r}{z} \varphi\\
  r & \cceq \phi \mid \varphi? \mid r+r \mid r \conc r \mid r^*
\end{align*}
where $p \in P$, $z \in \Var$, and where $\phi$ stands for arbitrary propositional formulas over $P$. We use the abbreviations~$\ttrue = p \vee \neg p$ and $\ffalse = p \wedge \neg p$ for some atomic proposition~$p$. The regular expressions have two types of atoms: propositional formulas~$\phi$ over the atomic propositions and tests~$\varphi?$, where $\varphi$ is again a $\pldl$ formula. Note that the semantics of the propositional atom~$\phi$ differ from the semantics of the test~$\phi?$: the former consumes an input letter, while tests do not. This is why both types of atoms are allowed.

The set of subformulas of $\varphi$ is denoted by $\cl(\varphi)$. Regular expressions are not subformulas, but the formulas appearing in the tests are, e.g., we have $\cl(\ddiamondle{\halfthinspace p?\conc q}{x}p') = \set{\halfthinspace p, p', \ddiamondle{\halfthinspace p?\conc q}{x}p'}$. The size~$\card{\varphi}$ of $\varphi$ is the sum of $\card{\cl(\varphi)}$ and the sum of the lengths of the regular expressions appearing in $\varphi$ (counted with multiplicity). 

We define $\vardiamond( \varphi ) = \{ z\in \Var \mid \ddiamondle{r}{z}\psi \in
\cl( \varphi) \}$ to be the set of variables parameterizing diamond-operators in
$\varphi$,  $\varbox( \varphi ) = \{ z\in \Var \mid \bboxle{r}{z} \psi \in \cl( \varphi)  \} $
to be the set of variables parameterizing box-operators in $\varphi$, and set
$\var( \varphi ) = \vardiamond( \varphi ) \cup \varbox( \varphi )$. Usually, we will denote variables in $\vardiamond( \varphi )$ by $x$ and variables in $\varbox( \varphi )$ by $y$, if $\varphi$ is clear from the context. A
formula~$\varphi$ is variable-free, if $\var( \varphi ) = \emptyset$.

The semantics of $\pldl$ is defined inductively with respect to $w = w_0 w_1 w_2 \cdots \in (\pow{P})^\omega$, a position~$n \in \nats$, and a variable valuation~$\alpha \colon \Var \rightarrow \nats$ via
\begin{itemize}

\item $(w, n, \alpha) \models p$ if $p \in w_n$,

\item $(w, n, \alpha) \models \neg p$ if $p \notin w_n$,

\item $(w, n, \alpha) \models \varphi_0 \wedge \varphi_1$ if $(w, n, \alpha) \models \varphi_0$ and $(w, n, \alpha) \models \varphi_1$,

\item $(w, n, \alpha) \models \varphi_0 \vee \varphi_1$ if $(w, n, \alpha) \models \varphi_0$ or $(w, n, \alpha) \models \varphi_1$,

\item $(w, n, \alpha) \models \ddiamond{r}\varphi$ if there exists $j \in \nats$ s.t.\ $(n, n+j) \in \Rexp(r, w, \alpha)$ and $(w, n+j, \alpha) \models \varphi$,

\item $(w, n, \alpha) \models \bbox{r}\varphi$ if for all $j \in \nats$ with $(n, n+j) \in \Rexp(r, w, \alpha)$ we have $(w, n+j, \alpha) \models \varphi$,

\item $(w, n, \alpha) \models \ddiamondle{r}{z}\varphi$ if there exists $0 \le j \le \alpha(z)$ s.t.\ $(n, n+j) \in \Rexp(r, w, \alpha)$ and $(w, n+j, \alpha) \models \varphi$, and

\item $(w, n, \alpha) \models \bboxle{r}{z}\varphi$ if for all $0 \le j \le \alpha(z)$ with $(n, n+j) \in \Rexp(r, w, \alpha)$ we have $(w, n+j, \alpha) \models \varphi$.

\end{itemize}
The relation~$\Rexp(r,w,\alpha) \subseteq \nats\times\nats$ contains all pairs~$(m,n)$ such that $w_m \cdots w_{n-1}$ matches $r$ and is defined inductively by 
\begin{itemize}
\item $\Rexp(\phi,w,\alpha) = \set{(n, n+1) \mid w_n \models \phi}$ for propositional~$\phi$,
\item $\Rexp(\theta?,w,\alpha) = \set{(n, n) \mid (w, n, \alpha) \models \theta}$,
\item $\Rexp(r_0 + r_1, w, \alpha) = \Rexp(r_0, w, \alpha) \cup \Rexp(r_1, w, \alpha)$,
\item $\Rexp(r_0 \conc r_1, w, \alpha) = \set{(n_0, n_2) \mid \exists n_1 \text{ s.t. }(n_0,n_1)\in \Rexp(r_0, w, \alpha) \text{ and } (n_1, n_2) \in \Rexp(r_1, w, \alpha)}$, and 
\item $\Rexp(r^*, w, \alpha) = \set{(n,n) \mid n\in\nats} \cup \set{(n_0, n_{k+1}) \mid \exists n_1, \ldots, n_{k} \text{ s.t. } (n_j, n_{j+1}) \in \Rexp(r, w, \alpha) \text{ for all } 0 \le j \le k}$.
\end{itemize}
We write $(w,\alpha) \models \varphi$ for $(w, 0, \alpha) \models \varphi$ and say that $w$ is a model of $\varphi$ with respect to $\alpha$.

\begin{example}\label{ex_pldl} \hfill
\begin{itemize}

\item The formula~$\chi_{\infty p} \coloneq \bbox{\ttrue^*}\ddiamond{\ttrue^*}p$ expresses that $p$ holds true infinitely often. 

\item In general, every $\pltl$ formula~\cite{AlurEtessamiLaTorrePeled01} (and thus every $\ltl$ formula) can be translated into $\pldl$, e.g., $\F_{\le x} p$ is expressible as $\ddiamondle{\ttrue^*}{x} p$ and $ p \U q$ as  $\ddiamond{\halfthinspace p^*}q$ or $\ddiamond{\halfthinspace p^*q}\ttrue$.

\item  The formula~$\bbox{\ttrue^*}(\mathit{req} \rightarrow \ddiamond{(\ttrue \conc\ttrue)^*}\mathit{resp})$ requires that every request (a position where $\mathit{req}$ holds) is followed by a response (a position where $\mathit{resp}$ holds) after an even number of steps. 
Note that the implication is not part of $\pldl$, but it can (here) be replaced by a disjunction.

\end{itemize}
\end{example}

As usual for parameterized temporal logics, the use of variables has to be
restricted: bounding diamond- and box-operators by the same variable leads
to an undecidable satisfiability problem (cf.~\cite{AlurEtessamiLaTorrePeled01}).

\begin{definition}
\label{def_wellformedformula}
A $\pldl$ formula~$\varphi$ is well-formed, if $\vardiamond( \varphi ) \cap \varbox( \varphi ) =
\emptyset$.
\end{definition}
In the following, we only consider well-formed formulas and drop the qualifier \quot{well-formed} whenever possible.

Note that we define $\pldl$ formulas to be in negation normal form.
Nevertheless, we can define the negation of a formula using dualities.

\begin{lemma}
\label{lemma_pldlnegation}
For every $\pldl$ formula~$\varphi$ there exists an efficiently constructible (not necessarily well-formed)
$\pldl$
formula~$\neg \varphi$ s.t.\
\begin{enumerate}
\item $(w,n,\alpha)\models \varphi$ if and only if $(w,n,
\alpha) \not\models \neg \varphi$, and
\item $\card{\neg \varphi} =  \card{\varphi}$.
\end{enumerate}
\end{lemma}

\begin{proof}
We construct $\neg \varphi$ by structural induction over $\varphi$ using the dualities of the operators:\medskip

\hspace{-.6cm}\begin{minipage}[b]{.44\textwidth}
	\begin{itemize}
		\item $\neg (p) = \neg p$
		\item $\neg (\varphi_0 \wedge \varphi_1) = (\neg \varphi_0) \vee (\neg \varphi_1)$
		\item $\neg (\ddiamond{r}\varphi) = \bbox{r}\neg \varphi$
		\item $\neg (\ddiamondle{r}{x}\varphi) = \bboxle{r}{x}\neg \varphi$
	\end{itemize}
	\end{minipage}
\begin{minipage}[b]{.44\textwidth}
	\begin{itemize}
		\item $\neg (\neg p) = p$
		\item $\neg (\varphi_0 \vee \varphi_1) = (\neg \varphi_0) \wedge (\neg \varphi_1)$
		\item $\neg (\bbox{r}\varphi) = \ddiamond{r}\neg \varphi$
		\item $\neg (\bboxle{r}{y}\varphi) = \ddiamondle{r}{y}\neg \varphi$
	\end{itemize}
	\end{minipage}\medskip

The latter claim of Lemma~\ref{lemma_pldlnegation} follows from the definition of $\neg \varphi$ while the first one can be shown by a straightforward structural induction over $\varphi$.
\end{proof}

Note that negation does not necessarily preserve well-formedness, e.g., the negation of the well-formed formula~$\varphi_{\boxbox} = \bboxle{(\bboxle{p}{x}p)?}{x}p$ is $\ddiamondle{(\bboxle{p}{x}p)?}{x}\neg p$, which is not well-formed. 

We consider the following fragments of $\pldl$. Let $\varphi$ be a $\pldl$ formula: 
\begin{itemize}
	\item $\varphi$ is an $\ldl$ formula~\cite{GiacomoVardi13}, if $\varphi$ is variable-free,
	\item $\varphi$ is a $\pldldiamond$ formula, if $\varbox(\varphi) = \emptyset$, and
	\item $\varphi$ is a $\pldlbox$ formula, if $\vardiamond(\varphi) = \emptyset$ and if $\neg \varphi$ is a $\pldldiamond$ formula\footnote{The definition of $\pldlbox$ in the conference version~\cite{FaymonvilleZimmermann14} is slightly too inclusive, because it contains the formula $\varphi_{\boxbox}$. This is problematic, as we have to require the negation of a $\pldlbox$ formula to be a $\pldldiamond$ formula.}. Note that this implies that a $\pldlbox$ formula cannot have parameterized subformulas in a test.	
\end{itemize}

Every $\ldl$, $\pldldiamond$, and $\pldlbox$ formula is well-formed by definition. As satisfaction of $\ldl$ formulas is independent of valuations, we write $(w, n) \models \varphi$ and $w \models \varphi$ instead of $(w, n, \alpha) \models \varphi$ and $(w, \alpha) \models \varphi$, respectively, if $\varphi$ is an $\ldl$ formula. 

$\ldl$ is as expressive as $\omega$-regular languages, which can be proven by a straightforward translation of ETL$_f$~\cite{VardiWolper94}, which expresses exactly the $\omega$-regular languages, into $\ldl$, and by a translation of $\ldl$ into Büchi automata. 

\begin{theorem}[\cite{Vardi11}]
Let $L \subseteq (\pow{P})^\omega$. The following are effectively equivalent:
\begin{enumerate}
	\item $L$ is $\omega$-regular.
	\item There exists an $\ldl$ formula $\varphi$ such that $L = \set{w \in (\pow{P})^\omega \mid w \models \varphi }$.
\end{enumerate}
\end{theorem}

A simple, but very useful property of $\pldl$ is the monotonicity of the
parameterized operators: increasing (decreasing) the values of parameters bounding diamond-operators (box-operators) preserves satisfaction.  

\begin{lemma}
\label{lemma_monotonicity}
Let $\varphi$ be a $\pldl$ formula and let $\alpha$ and
$\beta$ be variable valuations satisfying $\beta ( x) \ge \alpha ( x )$ for
every $x \in \vardiamond( \varphi)$ and $\beta ( y) \le \alpha ( y )$ for every $y \in \varbox( \varphi)$. If $(w, \alpha) \models \varphi$, then $(w, \beta) \models
\varphi$.
\end{lemma}

The previous lemma allows us to eliminate parameterized box-operators when asking for the existence of a variable valuation satisfying a formula. 

\begin{lemma}
	\label{lemma_removeboxes}
For every $\pldl$ formula~$\varphi$ there is an efficiently constructible $\pldldiamond$ formula~$\varphi'$ whose size is at most the size of $\varphi$ such that
\begin{enumerate}
	
	\item for every $\alpha$ there is an $\alpha'$ such that for all $w$:  $(w, \alpha) \models \varphi$ implies $(w, \alpha') \models \varphi'$, and
	
	\item for every $\beta'$ there is a $\beta$ such that for all $w$: $(w, \beta') \models \varphi'$ implies $(w, \beta) \models \varphi$.	

\end{enumerate}
\end{lemma}
 
\begin{proof}
For each $r$, we construct a test~$\hat{r}$ such that $\Rexp(r, w, \alpha) \cap \set{ (n,n) \mid n\in\nats} = \Rexp(\hat{r}, w, \alpha)$ for every $w$ and every $\alpha$. Then, $\bboxle{r}{y}\psi$ and $\bbox{\hat{r}}\psi$ are equivalent,  provided we have $\alpha(y) = 0$, which in combination with monotonicity is sufficient to prove our claim. We apply the following rewriting rules (in the given order) to $r$:
\begin{enumerate}
	\item Replace every subexpression $r'^*$ by $\ttrue?$, until no longer applicable.
	\item Replace every subexpression $\phi \conc r'$ or $r' \conc \phi$ by $\ffalse?$ and replace every subexpression $\phi + r'$ or $r' + \phi$ by $r'$, where $\phi$ is a propositional formula, until no longer applicable.
	\item Replace every subexpression $\theta_0? + \theta_1?$ by $(\theta_0 \vee \theta_1)?$ and replace every subexpression $\theta_0?\conc \theta_1?$ by $(\theta_0 \wedge \theta_1)?$, until no longer applicable.
\end{enumerate}
After step~2, $r$ contains no  iterations and no propositional atoms unless the expression itself is one. In the former case, applying the last two rules yields a regular expression, which is a single test, denoted by $\hat{r}$. In the latter case, we define $\hat{r} = \ffalse?$. 

Each rewriting step preserves the intersection~$\Rexp(r, w, \alpha) \cap \set{ (n,n) \mid n\in\nats}$. As $\hat{r}$ is a test, we conclude $\Rexp(r, w, \alpha) \cap \set{ (n,n) \mid n\in\nats} = \Rexp(\hat{r}, w, \alpha)$ for every $w$ and every $\alpha$. Note that $\hat{r}$ can be efficiently computed from $r$ and its size is at most the size of $r$.
Now, replace every subformula~$\bboxle{r}{y}\psi$ of $\varphi$ by $\bbox{\hat{r}}\psi$ and denote the formula obtained by $\varphi'$, which is a $\pldldiamond$ formula that is efficiently constructible and has the desired size. 

Given an $\alpha$, we define $\alpha'$ by $\alpha'(z) = 0$ if $z \in \varbox(\varphi)$, and $\alpha'(z) = \alpha(z)$ otherwise. If $(w, \alpha) \models \varphi$, then $(w, \alpha') \models \varphi$ due to monotonicity. By construction of $\varphi'$, we also have $(w, \alpha')\models \varphi'$.
On the other hand, if $(w, \beta') \models \varphi'$ for some $\beta'$, then $(w, \beta) \models \varphi'$ as well, where $\beta(z) = 0$, if $z \in \varbox(\varphi)$, and $\beta(z) = \beta'(z)$ otherwise. By construction of $\varphi'$, we conclude $(w, \beta) \models \varphi$. 
\end{proof}

\subsection{The Alternating Color Technique and LDL${_{cp}}$}
\label{subsec_altcolor}
In this subsection, we repeat the alternating color technique~\cite{KupfermanPitermanVardi09}, which was introduced by Kupferman et al.\ to solve the model checking and the realizability problem for $\prompt$, amongst others. Let $p \notin P$ be a fresh proposition and define $P' = P\cup \set{p}$. We think of words in $(\pow{P'})^\omega$ as colorings of words in $(\pow{P})^\omega$, i.e., $w' \in (\pow{P'})^\omega$ is a coloring of $w \in (\pow{P})^\omega$, if we have ${w_n}' \cap P = w_n$ for every position~$n$. Furthermore, $n$ is a changepoint, if $n= 0$ or if the truth value of $p$ differs at positions~$n-1$ and $n$. A block is a maximal infix that has exactly one changepoint, which is at the first position of the infix. By maximality, this implies that the first position after a block is a changepoint. Let $k\ge 1$. We say that $w'$ is $k$-bounded, if every block has length at most $k$, which implies that $w'$ has infinitely many changepoints. Dually, $w'$ is $k$-spaced, if it has infinitely many changepoints and every block has length at least $k$. 

The alternating color technique replaces every parameterized diamond-oper\-ator~$\ddiamondle{r}{x}\psi$ by an unparameterized one that requires the formula~$\psi$ to be satisfied within at most one color change. To this end, we introduce a changepoint-bounded variant~$\ddiamondcp{\cdot}\!$ of the diamond-operator. Since we need the dual operator~$\bboxcp{\!\cdot\!}\!$ to allow for negation via dualization, we introduce it here as well:
\begin{itemize}
	\item $(w, n, \alpha) \models \ddiamondcp{r}\psi$ if there exists a $j \in \nats$ s.t.\ $(n, n+j) \in \Rexp(r, w, \alpha)$, $w_n \cdots w_{n+j-1}$ contains at most one changepoint, and $(w, n + j, \alpha) \models \psi$, and

	\item $(w, n, \alpha) \models \bboxcp{r}\psi$ if for all $j \in \nats$ with $(n, n+j) \in \Rexp(r, w, \alpha)$ and where $w_n \cdots w_{n+j-1}$ contains at most one changepoint we have $(w, n + j, \alpha) \models \psi$.

\end{itemize}

We denote the logic obtained by disallowing parameterized operators, but allowing changepoint-bounded operators, by $\ldlt$. Note that the semantics of $\ldlt$ formulas are independent of variable valuations. Hence, we drop them from our notation for the satisfaction relation~$\models$ and the relation $\Rexp$. Also, Lemma~\ref{lemma_pldlnegation} can be extended to $\ldlt$ by adding the rules $\neg(\ddiamondcp{r}\psi) = \bboxcp{r}\neg \psi$ and $\neg(\bboxcp{r}\psi) = \ddiamondcp{r}\neg \psi$ to the proof.

Now, we are ready to introduce the alternating color technique. Given a $\pldldiamond$ formula~$\varphi$, let $\rel{\varphi}$ be the formula obtained by inductively replacing every subformula~$\ddiamondle{r}{x}\psi$ by $\ddiamondcp{\rel{r}}\rel{\psi}$, i.e., we replace the parameterized diamond-operator by a changepoint-bounded one. Note that this replacement is also performed in the regular expressions, i.e., $\rel{r}$ is the regular expression obtained by applying the replacement to every test~$\theta?$ in $r$. 

Given a $\pldldiamond$ formula~$\varphi$ let $c(\varphi) = \rel{\varphi} \wedge \chi_{\infty p} \wedge \chi_{\infty \neg p}$ (cf.~Example~\ref{ex_pldl}),
which is an $\ldlt$ formula and only linearly larger than $\varphi$. On $k$-bounded and $k$-spaced colorings of $w$ (for a suitable $k$) there is an equivalence between $\varphi$ and $c(\varphi)$. The proof is similar to the original one for $\prompt$~\cite{KupfermanPitermanVardi09}.

\begin{lemma}[cf. Lemma~2.1 of \cite{KupfermanPitermanVardi09}]
\label{lemma_altcolor}
Let $\varphi$ be a $\pldldiamond$ formula and let $w \in (\pow{P})^\omega$.
\begin{enumerate}
\item\label{lemma_alternatingcolor_pldltoldl} 
If $(w, \alpha) \models \varphi$, then $w' \models c(\varphi)$ for every $k$-spaced coloring $w'$ of $w$, where $k = \max_{x \in \var(\varphi)}\alpha(x)$.

\item\label{lemma_alternatingcolor_ldltopldl}
Let $k \in \nats$. If $w'$ is a $k$-bounded coloring of $w$ with $w' \models c(\varphi)$, then $(w, \alpha) \models \varphi$, where $\alpha(x) = 2k$ for every $x$.
\end{enumerate}
\end{lemma}
 
\section{From LDL$\mathbf{_{cp}}$ to Alternating Büchi Automata}
\label{sec_automata}
In this section, we show how to translate $\ldlt$ formulas into alternating Büchi word automata with linearly many states, but possibly with an exponential number of transitions, using an inductive bottom-up approach. These automata allow us to use automata-based constructions to solve the model checking and the realizability problem for $\pldl$ via the alternating color technique which links $\pldl$ and $\ldlt$. Since these problems are shown to be complete for the complexity classes $\pspace$ and $\twoexp$, which allow us to construct the automata (on-the-fly), the potentially exponential number of transitions is not an issue. 

An alternating Büchi automaton~$\aut = (Q,\Sigma,q_0,\delta, F)$ consists of a finite set~$Q$ of states, an alphabet~$\Sigma$, an initial state~$q_0 \in Q$, a transition function~$\delta \colon Q \times \Sigma \to \bplus(Q)$, and a set~$F \subseteq Q$ of accepting states. 
Here, $\bplus(Q)$ denotes the set of positive boolean combinations over $Q$, which contains in particular the formulas $\ttrue$ (true) and $\ffalse$ (false).  

A run of $\aut$ on $w = w_0 w_1 w_2 \cdots \in \Sigma^\omega$ is a directed graph~$\rho = (V, E)$ where $V \subseteq Q \times \nats$ and $((q,n),(q',n')) \in E$ implies $n' = n +1$ such that the following two conditions are satisfied: $(q_0, 0) \in V$ and for all $(q, n) \in V$: $\suc{
\rho}{(q,n)} \models \delta(q, w_n)$. Here $\suc{\rho}{(q,n)}$ denotes the set of successors of $(q,n)$ in $\rho$ projected to $Q$. A run~$\rho$ is accepting if all infinite paths (projected to $Q$) through $\rho$ visit $F$ infinitely often. The language~$L(\aut)$ contains all $w \in \Sigma^\omega$ that have an accepting run of $\aut$.

\begin{theorem}
	\label{theorem_autconstruction}
For every $\ldlt$ formula~$\varphi$, there is an alternating Büchi automaton~$\aut_\varphi$ with linearly many states (in $\card{\varphi}$) and $L(\aut_\varphi) = \set{w \in (\pow{P'})^\omega \mid w \models \varphi }$.
\end{theorem}

To prove the theorem, we inductively construct automata~$\aut_\psi$ for every subformula~$\psi \in \cl(\varphi)$ satisfying $L(\aut_\psi) = \set{w \in (\pow{P'})^\omega \mid w \models \psi}$. 

The automata for atomic formulas are straightforward and depicted in Figures~\ref{fig_atomicaut}(a) and (b). To improve readability, we allow propositional formulas over $P'$ as transition labels: a formula~$\phi$ stands for all sets $A \in \pow{P'}$ with $A \models \phi$. 

Furthermore, given automata~$\aut_{\psi_0}$ and $\aut_{\psi_1}$, using a standard construction, we can build the automaton~$\aut_{\psi_0 \vee \psi_1}$ by taking the disjoint union of the two automata, adding a new initial state~$q_0$ with $\delta(q_0, A) = \delta^0(q_0^0, A) \vee \delta^1(q_0^1, A)$. Here, $q_0^i$ is the initial state and $\delta^i$ is the transition function of $\aut_{\psi_i}$. The automaton~$\aut_{\psi_0 \wedge \psi_1}$ is defined similarly, the only difference being $\delta(q_0, A) = \delta^0(q_0^0, A) \wedge \delta^1(q_0^1, A)$.

\begin{figure}[h]
	\begin{center}
\begin{tikzpicture}[scale = .72,
	every initial by arrow/.style={initial text=,-stealth, thick},
	every state/.style={thick}]

\node			at (-.8,1.75)	{(a)};
\node			at (3.3,1.75)	{(b)};
\node			at (7.5,1.75)	{(c)};

\node[state, initial] (q) at (0,.0)	{};
\node		(a) at (1.7,-.9)	{$\ffalse$};
\node		(b) at (1.7, .9)	{$\ttrue$};

\path[thick,-stealth]
(q) edge[bend right] node[below, yshift =-.15cm]{$\neg p$} (a)
(q) edge[bend left] node[above, yshift =.15cm]{$p$} (b);

\node[state, initial] (qp) at (4,0)	{};
\node						(ap) at (5.7,-.9)	{$\ffalse$};
\node						(bp) at (5.7, .9)	{$\ttrue$};

\path[thick,-stealth]
(qp) edge[bend right] node[below, yshift =-.15cm]{$p$} (ap)
(qp) edge[bend left] node[above, yshift =.15cm]{$\neg p$} (bp);

\node[state, initial, accepting] (e) at (8, 0) {};
\node[state, accepting] (b) at (10.5, .9) {};
\node[state, accepting] (y) at (10, -.9) {};
\node[state, accepting] (yb) at (13, -.9) {};
\node[state, accepting] (by) at (13.5, .9) {};
\node[state] (s) at (15.5, 0) {};

\path[thick, -stealth]
(e) edge[bend left] node[above]{$p$} (b)
(e) edge[bend right] node[below]{$\neg p$} (y)
(b) edge[in=210,out=240,loop] node[left]{$\,\, p$} ()
(y) edge[in=30,out=60,loop] node[right,xshift=-.2cm]{$\,\, \neg p$} ()
(b) edge			node[above]{$\neg p$} (by)
(y) edge			node[below]{$p$} (yb)
(yb) edge[in=30,out=60,loop] node[right]{$\,\, p$} ()
(by) edge[in=210,out=240,loop] node[left]{$\,\, \neg p$} ()
(yb) edge[bend right] node[below]{$\neg p$} (s)
(by) edge[bend left] node[above]{$p$} (s)
(s) edge[loop above] node[above] {$\ttrue$} ()
;

\end{tikzpicture}
\caption{The automata~$\aut_{p}$ (a), $\aut_{\neg p}$ (b), and $\aut_{cp}$ (c), which tracks changepoints.}
\label{fig_atomicaut}
	\end{center}
\end{figure}
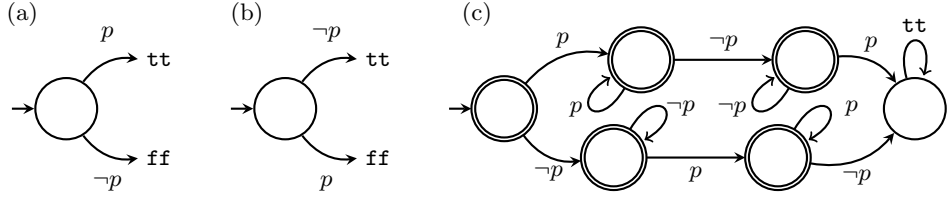

It remains to consider temporal formulas, e.g., $\ddiamond{r}\psi$. First, we turn the regular expression~$r$ into an automaton~$\aut_r$. Recall that tests do not process input letters. Hence, we disregard the tests when defining the transition function, but we label states at which the test has to be executed, by this test. We adapt the Thompson construction~\cite{Thompson68} to turn~$r$ into $\aut_r$, i.e., we obtain an $\epsilon$-NFA. Then, we show how to combine $\aut_r$ with the automaton~$\aut_\psi$ and the automata~$\aut_{\theta_1}, \ldots, \aut_{\theta_k}$, where $\theta_1?, \ldots, \theta_k?$ are the tests occurring in $r$. The $\epsilon$-transitions introduced by the Thompson construction are removed during the construction, since alternating automata do not allow them. During this construction, we also ensure that the transition relation takes tests into account by introducing universal transitions that lead from a state marked with $\theta_j?$ into the corresponding automaton~$\aut_{\theta_j}$.

An $\epsilon$-NFA with markings $\aut = (Q, \Sigma, q_0, \delta, C, \marking)$ consists of a finite set~$Q$ of states, an alphabet~$\Sigma$, an initial state~$q_0 \in Q$, a transition function~$\delta \colon Q \times \Sigma\cup\set{\epsilon} \rightarrow \pow{Q}$, a set~$C$ of final states ($C$, since we use them to concatenate automata), and a partial marking function~$\marking$, which assigns to some states~$q \in Q$ an $\ldlt$ formula~$\marking(q)$.
We write $q \xrightarrow{a} q'$, if $q' \in \delta(q, a)$ for $a \in \Sigma \cup \set{\epsilon}$. An $\epsilon$-path~$\pi$ from $q$ to $q'$ in $\aut$ is a sequence~$\pi = q_1 \cdots q_k$ of $k \ge 1$ states with $q =q_1 \xrightarrow{\epsilon} \cdots \xrightarrow{\epsilon} q_k = q'$. The set of all $\epsilon$-paths from $q$ to $q'$ is denoted by $\Pi(q, q')$ and $\marking(\pi) = \set{\marking(q_i) \mid 1 \le i \le k}$  is the set of markings visited by $\pi$.

A run of $\aut$ on $w_0 \cdots w_{n-1} \in \Sigma^*$ is a sequence~$q_0 q_1 \cdots q_n$ of states such that for every $i$ in the range~$0 \le i \le n-1$ there is a state~$q_{i}'$ reachable from $q_i$ via an $\epsilon$-path~$\pi_{i}$ and with $q_{i+1} \in \delta(q_{i}', w_{i})$. The run is accepting if there is a $q_{n}' \in C$ reachable from $q_n$ via an $\epsilon$-path~$\pi_n$.  This slightly unusual definition (but equivalent to the standard one) simplifies our reasoning below. Also, the definition is oblivious to the marking.

We begin by defining the automaton~$\aut_r$ by induction over the structure of $r$ as depicted in Figure~\ref{fig_autr}. Note that the automata we construct have no outgoing edges leaving the unique final state and that we mark some states with tests~$\theta_j?$ (denoted by labeling states with the test).

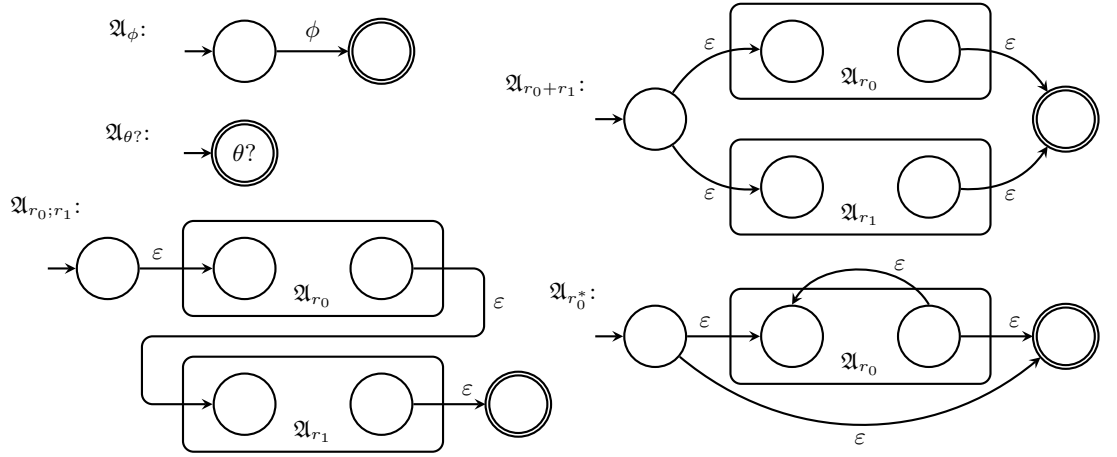
\begin{figure}[h!]
	\begin{center}
\begin{tikzpicture}[
	scale=.9,
	every initial by arrow/.style={initial text=,-stealth, thick},
	every state/.style={thick}]

\node at (.3,.3) {$\aut_\phi$:};
\node at (.3, -1.2) {$\aut_{\theta?}$:};
\node at (6.4,-.5) {$\aut_{r_0 + r_1}$:};
\node at (-0.9,-2.3) {$\aut_{r_0; r_1}$:};
\node at (6.8, -3.6) {$\aut_{r_0^*}$:};

\node[state, initial] (q) at (2,0)	{};
\node[state, accepting]		(a) at (4,0)	{};

\path[thick,-stealth]
(q) edge node[above]{$\phi$} (a);


\node[state, initial, accepting] (qp) at (2,-1.5)	{$\theta?$};


\draw[rounded corners, thick]  (9.1,.7) rectangle (12.9,-.7);
\draw[rounded corners, thick]  (9.1,-1.3) rectangle (12.9,-2.7);

\node		at (11,-.4) {$\aut_{r_0}$};
\node		at (11,-2.4) {$\aut_{r_1}$};

\node[state, initial] (u0) at (8,-1) {};
\node[state] (u1) at (10,0) {};
\node[state] (u2) at (10,-2) {};
\node[state] (u3) at (12,0) {};
\node[state] (u4) at (12,-2) {};
\node[state, accepting] (u5) at (14,-1) {};

\path[thick, -stealth]
(u0) edge[bend left] node[above] {$\epsilon$} (u1)
(u0) edge[bend right] node[below] {$\epsilon$} (u2)
(u3) edge[bend left] node[above] {$\epsilon$} (u5)
(u4) edge[bend right] node[below] {$\epsilon$} (u5);


\draw[rounded corners, thick]  (1.1,-2.5) rectangle (4.9,-3.9);
\draw[rounded corners, thick]  (1.1,-4.5) rectangle (4.9,-5.9);

\node		at (3,-3.6) {$\aut_{r_0}$};
\node		at (3,-5.6) {$\aut_{r_1}$};

\node[state, initial] (u0) at (0,-3.2) {};
\node[state] (u1) at (2,-3.2) {};
\node[state] (u2) at (2,-5.2) {};
\node[state] (u3) at (4,-3.2) {};
\node[state] (u4) at (4,-5.2) {};
\node[state, accepting] (u5) at (6,-5.2) {};

\path[thick, -stealth]
(u0) edge[] node[above, near start] {$\epsilon$} (u1)
(u4) edge[] node[above, near end] {$\epsilon$} (u5);

\draw[thick, rounded corners, -stealth]
(u3.east) -- (5.5,-3.2) -- node[right] {$\epsilon$}  (5.5, -4.2) -- (.5, -4.2) -- (.5, -5.2) -- (u2.west);


\draw[rounded corners, thick]  (9.1,-3.5) rectangle (12.9,-4.9);
\node		at (11,-4.6) {$\aut_{r_0}$};

\node[state, initial] (c0) at (8,-4.2) {};
\node[state] (c1) at (10,-4.2) {};
\node[state] (c2) at (12,-4.2) {};
\node[state, accepting] (c3) at (14,-4.2) {};

\path[thick, -stealth]
(c0) edge[bend right = 40] node[below] {$\epsilon$} (c3)
(c2.north) edge[bend right = 60] node[above, near start] {$\epsilon$} (c1.north)
(c0) edge node[above, near start] {$\epsilon$} (c1)
(c2) edge node[above, near end] {$\epsilon$} (c3);

\end{tikzpicture}
\caption{The inductive definition of $\aut_r$ via the Thompson construction.}
\label{fig_autr}
	\end{center}
\end{figure}

\begin{lemma}
\label{lemma_autrcorrectness}
Let $w = w_0 w_1 w_2 \cdots \in (\pow{P'})^\omega$ and let $w_0 \cdots w_{n-1}$ be a (possibly empty, if $n=0$) prefix of $w$. The following two statements are equivalent:
\begin{enumerate}
	\item $\aut_r$ has an accepting run on $w_0 \cdots w_{n-1}$ with $\epsilon$-paths $\pi_0, \ldots, \pi_n$ such that $w_{i}w_{i+1}w_{i+2} \cdots \models \bigwedge \marking(\pi_i)$ for every $i$ in the range~$0 \le i \le n$ .

\item $(0,n) \in \Rexp(r,w)$.

\end{enumerate}
\end{lemma}
 

Fix $\psi$ and $r$ (with tests~$\theta_1?, \ldots, \theta_k?$) and let
$\aut_r = (Q^r, \pow{P'}, q_0^r, \delta^r, C^r, \marking)$, $\aut_{\psi} = (Q', \pow{P'}, q_0', \delta', F')$, and $\aut_{\theta_j} = (Q^j, \pow{P'}, q_0^j, \delta^j, F^j)$ for $j = 1, \ldots, k$ be the corresponding automata, which we assume to have pairwise disjoint sets of states. Next, we show how to construct  $\aut_{\ddiamond{r}\psi}$, $ \aut_{\bbox{r}\psi}$, $ \aut_{\ddiamondcp{r}\psi}$, and $\aut_{\bboxcp{r}\psi}$.

We begin with $\ddiamond{r}\psi$ and define 
\[\aut_{\ddiamond{r}\psi} = (Q^r \cup Q' \cup Q^1 \cup \cdots \cup Q^k, \pow{P'}, q_0^r, \delta, F' \cup F^1 \cup \cdots \cup F^k)\] with
\[\delta(q, A) = \begin{cases}
	\delta'(q, A) 		&\text{if $q \in Q'$},\\
	\delta^j(q, A) 		&\text{if $q \in Q^j$},\\
	\bigvee_{q' \in Q^r \setminus C^r}\bigvee_{\pi \in \Pi(q,q')} \bigvee_{p \in \delta^r (q', A)} (p \wedge \bigwedge_{\theta_j \in \marking(\pi)} \delta^j(q_0^j, A))&\\
	 \hspace{4.cm}\vee &\text{if $q \in Q^r$}.\\
	\bigvee_{q' \in C^r}\bigvee_{\pi \in \Pi(q,q')} (\delta'(q_0', A) \wedge \bigwedge_{\theta_j \in \marking(\pi)} \delta^j(q_0^j, A))  &\\
	\end{cases}\]
So, $\aut_{\ddiamond{r}\psi}$ is the union of the automata for the regular expression, the tests, and for $\psi$ with a modified transition function. The transitions of the automata~$\aut_{\psi}$ and $\aut_{\theta_j}$ are left unchanged and the transition function for states in $Q^r$ is obtained by removing $\epsilon$-transitions. First consider the upper disjunct: it ranges disjunctively over all states~$p$ that are reachable via an initial $\epsilon$-path and an $A$-transition in the end. To account for the tests visited during the $\epsilon$-path (but not the test at $p$), we conjunctively add transitions that lead into the corresponding automata. The lower disjunct is similar, but ranges over $\epsilon$-paths that end in a final state, which requires the $A$ to be processed in $\aut_\psi$. Since we concatenate the automaton~$\aut_r$ with the automaton $\aut_\psi$, all edges leading into final states of $\aut_r$ are rerouted to the successors of the initial state of $\aut_\psi$. The tests along the $\epsilon$-paths are accounted for as in the first case. Finally, note that $Q^r$ does not contain any (Büchi) accepting states, i.e., every accepting run on $w$ has to leave $Q^r$ after a finite number of transitions. Since this requires transitions that would lead $\aut_r$ into a final state, we ensure the existence of a position $n$ such that $(0,n) \in \Rexp(r,w)$.

The definition of $\aut_{\bbox{r}\psi}$ is dual, which requires us to consider the negation of the tests: let $\aut_{\neg \theta_j}  = (Q^j, \pow{P'}, q_0^j, \delta^j, F^j)$ for $j = 1, \ldots, k$ be automata for the negation of the tests~$\theta_1?, \ldots, \theta_k?$ appearing in $r$. Recall that $\neg \theta_j$ always refers to the formula obtained by propagating the negation according to Lemma~\ref{lemma_pldlnegation}, and thus the automata for the negated tests can be obtained without using automata complementation. Furthermore, to construct $\aut_{\bbox{r}\psi}$, we remove $\epsilon$-paths of $\aut_r$ in a universal manner to account for the fact that the box-operator quantifies over all matches with $r$.  

Formally, we define 
\[\aut_{\bbox{r}\psi} = (Q^r \cup Q' \cup Q^1 \cup \cdots \cup Q^k, \pow{P'}, q_0^r, \delta, Q^r \cup F' \cup F^1 \cup \cdots \cup F^k)\] where
\[\delta(q, A) = \begin{cases}
	\delta'(q, A) 		&\text{if $q \in Q'$},\\
	\delta^j(q, A) 		&\text{if $q \in Q^j$},\\
	\bigwedge_{q' \in Q^r \setminus C^r}\bigwedge_{\pi \in \Pi(q,q')} \bigwedge_{p \in \delta^r (q', A)} (p \vee \bigvee_{\theta_j \in \marking(\pi)} \delta^j(q_0^j, A)) &\\
	\hspace{4.2cm}\wedge &\text{if $q \in Q^r$}.\\
	\bigwedge_{q' \in C^r}\bigwedge_{\pi \in \Pi(q,q')} (\delta'(q_0', A) \vee \bigvee_{\theta_j \in \marking(\pi)} \delta^j(q_0^j, A)) &\\
	\end{cases}\]
Note that we add $Q^r$ to the (Büchi) accepting states, since a path of a run on $w$ might stay in $Q^r$ forever, as it has to consider all $n$ with $(0,n) \in \Rexp(r,w)$.

For the changepoint-bounded operators, we have to modify $\aut_r$ to make it count color changes. Let $\aut_{cp} = (Q^{cp}, \pow{P'}, q_0^{cp}, \delta^{cp}, C^{cp})$ be the DFA depicted in Figure~\ref{fig_atomicaut}(c). We define the product of $\aut_r$ and $\aut_{cp}$ as 
\[\hat{\aut}_r = (\hat{Q}^r, \pow{P'}, \hat{q}_0^r, \hat{\delta}^r, \hat{C}^r, \hat{\marking}) \]
 where 
 \begin{itemize}
 	\item $\hat{Q}^r = Q^r \times Q^{cp}$, 
 	\item $\hat{q}_0^r = (q_0^r, q_0^{cp})$,
 	\item $\hat{\delta^r}((q,q'), A) = \set{(p,\delta^{cp}(q',A)) \mid p \in \delta^r(q, A)}$ if $A \not= \epsilon$, and
	$\hat{\delta^r}((q,q'), \epsilon) = \set{(p,q') \mid p \in \delta^r(q, A)}$,
	\item $\hat{C}^r = C^r \times C^{cp}$, and
	\item $\hat{\marking}(q,q') = \marking(q)$. 
 \end{itemize}
 Using this, we define $ \aut_{\ddiamondcp{r}\psi}$ as we defined $\aut_{\ddiamond{r}\psi}$, but using $\hat{\aut}_r$ instead of $\aut_r$. Similarly, $\aut_{\bboxcp{r}\psi}$ is defined as $\aut_{\bbox{r}\psi}$, but using $\hat{\aut}_r$ instead of $\aut_r$, which restricts the matches with $r$ recognized by $\aut_r$ to those that are within at most one changepoint.

It remains to prove that the construction is correct.

\begin{proof}[Proof of Theorem~\ref{theorem_autconstruction}]
First, we determine the size of $\aut_\varphi$. Boolean operations add one state while a temporal operator with regular expression~$r$ adds a number of states that is linear in the size of $r$ (which is its length), even when we take the intersection with the automaton checking for color changes. Note that we do not need to complement the automata~$\aut_{\theta_j}$ to obtain the $\aut_{\neg \theta_j}$, instead we rely on Lemma~\ref{lemma_pldlnegation}. Hence, the size of $\aut_\varphi$ is linear in the size of $\varphi$.

Thus, it remains to prove that $\aut_\varphi$ recognizes the models of $\varphi$. We proceed by induction over the structure of $\varphi$. The induction starts for atomic formulas and the induction steps for disjunction and conjunction are trivial, hence it suffices to consider the temporal operators. 

First, consider a subformula of the form~$\ddiamond{r}\psi$. If $w \models \ddiamond{r}\psi$, then there exists a position~$n$ such that $w_n w_{n+1} w_{n+2} \cdots \models \psi$ and $(0,n) \in \Rexp(r,w)$. Hence, due to Lemma~\ref{lemma_autrcorrectness}, there is an accepting run of $\aut_r$ on $w_0 \cdots w_{n-1}$ such that the tests visited during the run are satisfied by the appropriate suffixes of $w$. Thus, applying the induction hypothesis yields accepting runs of the test automata on these suffixes. Also, there is an accepting run of $\aut_\psi$ on $w_n w_{n+1} w_{n+2} \cdots$, again by induction hypothesis. These runs can be \quot{glued} together to obtain an accepting run of $\aut_{\ddiamond{r}\psi}$ on $w$.

For the other direction, let $\rho$ be an accepting run of $\aut_{\ddiamond{r}\psi}$ on $w$. Let $n \ge 0$ be the last level of $\rho$ that contains a state from $Q^r$. Such a level has to exist since states in $Q^r$ are not accepting and they have no incoming edges from states of the automata~$\aut_\psi$ and $\aut_{\theta_j}$ (the $\theta_j$ are the tests in $r$), but the initial state of $\aut_{\ddiamond{r}\psi}$ is in $Q^r$. Furthermore, $\aut_{\ddiamond{r}\psi}$ is non-deterministic when restricted to states in $Q^r \setminus C^r$. Hence, we can extract an accepting run of $\aut_r$ from $\rho$ on $w_0 \cdots w_{n-1}$ that additionally satisfies the requirements formulated in Statement~1 of Lemma~\ref{lemma_autrcorrectness}, due to the transitions into the test automata and applications of the induction hypothesis. Hence, we have $(0,n) \in \Rexp(r,w) $. Also, from the remainder of $\rho$ (levels greater or equal to $n$) we can extract an accepting run of $\aut_\psi$ on $w_n w_{n+1} w_{n+2} \cdots$. Hence, $w_n w_{n+1} w_{n+2} \cdots \models \psi$ by induction hypothesis. So, we conclude $w \models \ddiamond{r}\psi$.

The case for $\bbox{r}\psi$ is dual to the one for $\ddiamond{r}\psi$, while the cases for the change\-point-bounded operators~$\ddiamondcp{r}\psi$ and $\bboxcp{r}\psi$ are analogous, using the fact that $\aut_{cp}$ accepts words which have at most one changepoint. 
\end{proof}

The number of states of $\aut_\varphi$ is linear in $\card{\varphi}$, but it is not clear that $\aut_\varphi$ can be computed in polynomial time in $\card{\varphi}$, since, e.g., the transition functions of sub-automata of the form~$\aut_{\ddiamond{r}\psi}$ contain disjunctions that range over the set of $\epsilon$-paths. Here, it suffices to consider simple  paths, but even this restriction still allows for an exponential number of different paths. Fortunately, we do not need to compute $\aut_\varphi$ in polynomial time. It suffices to construct it on-the-fly in polynomial space, as this is sufficient for our  applications, which is clearly possible.

Furthermore, using standard constructions (e.g., \cite{MiyanoH84, Schewe09}), we can turn the alternating Büchi automaton~$\aut_\varphi$ into a non-deterministic Büchi automaton of exponential size and a deterministic parity automaton\footnote{The states of a parity automaton are colored by $\col\colon Q \rightarrow \nats$. It accepts a word~$w$, if it has a run~$q_0 q_1 q_2 \cdots$ on $w$ such that $\max\set{\col(q) \mid q_i = q \text{ for infinitely many i}} $ is even.} of doubly-exponential size with exponentially many colors.

Finally, the automata we construct are weak~\cite{MullerSS92}, i.e., every strongly connected component either has only accepting or only non-accepting states, which allows for improved translations into non-deterministic automata: the automata for the atomic formulas are weak and taking the union or intersection of two weak automata preserves weakness. Thus, consider the automata constructed for the temporal operators: the states of the automaton  for $r$ are either all accepting or all rejecting, and once this set of states is left to some automaton checking a subformula, it is never reentered. Hence, as these sub-automata are weak, the whole automaton is weak as well. However, our automata are not very weak~\cite{Rohde:1997:AAT:925874, GastinO01} (also known as linear), i.e., the automaton only has self-loops, but no non-trivial cycles, as the automata checking matches with $r$ might have cycles of arbitrary length. 

\section{Model Checking}
\label{sec_mc}
In this section, we consider the $\pldl$ model checking problem. A ($P$-labeled) transition system $\sys = (S, s_0, E, \ell)$ consists of a finite set~$S$ of states, an initial state~$s_0$, a left-total edge relation~$E \subseteq S \times S$, and a labeling~$\ell \colon S \rightarrow \pow{P}$. An initial path through $\sys$ is a sequence~$\pi = s_0 s_1 s_2\cdots$ of states satisfying $(s_n, s_{n+1}) \in E$ for every $n$. Its trace is defined as $\trace(\pi) = \ell(s_0) \ell(s_1) \ell(s_2) \cdots$. We say that $\sys$ satisfies a $\pldl$ formula~$\varphi$ with respect to a variable valuation~$\alpha$, if we have $(\trace(\pi), \alpha) \models \varphi$ for every initial path~$\pi$ of $\sys$.
The model checking problem asks, given a transition system~$\sys$ and a formula~$\varphi$, to determine whether there exists a variable valuation~$\alpha$ for which $\sys$ satisfies $\varphi$.%

\begin{theorem}
\label{thm_mc}
The $\pldl$ model checking problem is $\pspace$-complete.
\end{theorem}

To solve the $\pldl$ model checking problem, we first notice that we can restrict ourselves to $\pldldiamond$ formulas. Let $\varphi$ and $\varphi'$ be  defined as in Lemma~\ref{lemma_removeboxes}. Then, $\sys$ satisfies $\varphi$ with respect to some $\alpha$ if and only if $\sys$ satisfies $\varphi'$ with respect to some $\alpha'$. 

Our algorithm is similar to the one presented for $\prompt$ in \cite{KupfermanPitermanVardi09} and uses the alternating color technique. Recall that $p \notin P$ is the fresh atomic proposition used to specify the coloring and induces the blocks, maximal infixes with its unique changepoint at the first position. Let $G = (V, E, v_0, \ell, F)$  denote a colored Büchi graph consisting of a finite directed graph $(V, E)$, an initial vertex~$v_0$, a coloring function~$\ell \colon V \rightarrow \pow{\set{p}}$ labeling vertices by $p$ or not, and a set $F \subseteq V$ of accepting vertices. A path $v_0 v_1 v_2 \cdots $ through $G$ is pumpable, if all its blocks have at least one vertex that appears twice in this block. Furthermore, the path is fair, if it visits $F$ infinitely often. The pumpable non-emptiness problem asks, given a colored Büchi graph~$G$, whether it has a pumpable fair path starting in the initial vertex.

\begin{theorem}[\cite{KupfermanPitermanVardi09}]
The pumpable non-emptiness problem for colored Büchi graphs is $\nlogspace$-complete.
\end{theorem}

The following lemma reduces the $\pldldiamond$ model checking problem to the pumpable non-emptiness problem for colored Büchi graphs of exponential size. Given a non-deterministic Büchi automaton~$\aut = (Q, \pow{P'}, q_0, \delta, F)$ recognizing the models of $\neg \rel{ \varphi} \wedge \chi_{\infty p} \wedge \chi_{\infty \neg p}$ (note that $\rel{\varphi}$ is negated) and a transition system~$\sys = (S, s_0, E, \ell)$, define the product~$\aut \times \sys$ to be the colored Büchi~graph 
\[ \aut \times \sys = (Q \times S \times \pow{\set{p}}, E', (q_0, s_0, \emptyset), \ell', F \times S \times \pow{\set{p}})\]
 where
 \begin{itemize}
 	\item $((q, s, C),(q', s', C')) \in E'$ if and only if $(s,s')\in E$ and $q' \in \delta(q, \ell(s) \cup C)$, and
 	\item  $\ell'(q,s,C) = C$.
 \end{itemize}

Each initial path $(q_0, s_0, C_0)(q_1, s_1, C_1)(q_2, s_2, C_2)\cdots$ through $\aut \times \sys$ induces a coloring $(\ell (s_0) \cup C_0)(\ell (s_1) \cup C_1)(\ell (s_2) \cup C_2)\cdots$ of the trace of the path $s_0 s_1 s_2 \cdots$ through $\sys$. Furthermore, $q_0 q_1 q_2 \cdots$ is a run of $\aut$ on the coloring.

\begin{lemma}[cf.\ Theorem 4.2 of \cite{KupfermanPitermanVardi09}]
\label{lemma_pumppath}
$\sys$ does not satisfy $\varphi$ with respect to any $\alpha$ if and only if $\aut \times \sys$ has a pumpable fair path.
\end{lemma}

\begin{proof}
Let $\varphi$ not be satisfied by $\sys$ with respect to any $\alpha$, i.e., for every $\alpha$ there exists an initial path~$\pi$ through $\sys$ such that $(\trace(\pi), \alpha) \not\models \varphi$. Pick $\alpha^*$ such that $\alpha^*(z) = 2\cdot|Q| \cdot |S| +2$ for every $z$ and let $\pi^*$ be the corresponding path. Applying Item~\ref{lemma_alternatingcolor_ldltopldl} of Lemma~\ref{lemma_altcolor} yields $w \not \models c(\varphi)$ for every $|Q|\cdot|S|+1$-bounded coloring~$w$ of $\trace(\pi^*)$. Now, consider the unique $|Q|\cdot|S|+1$-bounded and $|Q|\cdot|S|+1$-spaced coloring~$w$ of $\trace(\pi^*)$ that starts with $p$ not holding true in the first position. As argued above, $w \not \models c(\varphi)$, and we have $w \models \chi_{\infty p} \wedge \chi_{\infty \neg p}$, as $w$ is bounded. Hence, $w \models \neg \rel{ \varphi} \wedge \chi_{\infty p} \wedge \chi_{\infty \neg p}$, i.e., there is an accepting run~$q_0 q_1 q_2 \cdots$ of $\aut$ in $w$. This suffices to show that $(q_0, \pi_0^*, w_0\cap \set{p}) (q_1, \pi_1^*, w_1\cap \set{p}) (q_2, \pi_2^*, w_2\cap \set{p}) \cdots$ is a pumpable fair path through $\aut \times \sys$, since every block has length~$|Q|\cdot|S|+1$. This implies the existence of a repeated vertex in every block, since there are exactly $|Q|\cdot|S|$ vertices of each color.

We now consider the other direction. Thus, assume $\aut \times \sys$ contains a pumpable fair path $(q_0, s_0, C_0)(q_1, s_1, C_1)(q_2, s_2, C_2)\cdots$, fix some arbitrary $\alpha$, and define $k = \max_{x \in \vardiamond(\varphi)}\alpha(x)$. There is a repetition of a vertex of $\aut \times \sys$ in every block, each of which can be pumped $k$ times. This path is still fair and induces a coloring~$w_k'$ of a trace~$w_k$ of an initial path of $\sys$. Since the run encoded in the first components is an accepting one on $w_k'$, we conclude that the coloring~$w_k'$ satisfies $\neg \rel{\varphi}$. Furthermore, $w_k'$ is $k$-spaced, since we pumped each repetition $k$ times.

Towards a contradiction assume we have $(w_k, \alpha) \models \varphi$. 
Applying Item~\ref{lemma_alternatingcolor_pldltoldl} of Lemma~\ref{lemma_altcolor} yields $w'_k \models c(\varphi)$, which contradicts $w'_k \models \neg \rel{\varphi}$. Hence, for every $\alpha$ we have constructed a path of $\sys$ whose trace does not satisfy $\varphi$ with respect to $\alpha$, i.e., $\sys$ does not satisfy $\varphi$ with respect to any $\alpha$.
\end{proof}

We can deduce an upper bound on valuations that satisfy a formula in a given transition system.

\begin{corollary}
	\label{cor_mcub}
If there is a valuation such that $\sys$ satisfies a $\pldldiamond$ formula $\varphi$, then there is also one that is bounded exponentially in $\size{\varphi}$ and linearly in  $\size{\sys}$.	
\end{corollary}

\begin{proof}
Let $\sys$ satisfy $\varphi$ with respect to $\alpha$, but not with the valuation~$\alpha^*$ with $\alpha^*(x) = 2\cdot|Q| \cdot |S| +2$ for all $x$. In the preceding proof, we constructed a pumpable fair path in $\aut \times \sys$ starting from this assumption. This contradicts Lemma~\ref{lemma_pumppath}, since $\sys$ satisfying $\varphi$ with respect to $\alpha$ is equivalent to $\aut \times \sys$ not having a pumpable fair path. Since $2\cdot|Q| \cdot |S| +2$ is exponential in $\size{\varphi}$ and linear in $\card{S}$, the result follows.
\end{proof}

A matching lower bound of $2^n$ can be proven by implementing a binary counter with $n$ bits using a formula of polynomial size in $n$. This already holds true for $\prompt$, as noted in~\cite{KupfermanPitermanVardi09}.

Now, we are able to prove the main result of this section: $\pldl$ model checking is $\pspace$-complete.

\begin{proof}[Proof of Theorem~\ref{thm_mc}]
$\pspace$-hardness follows directly from the $\pspace$-hard\-ness of the $\ltl$ model checking problem~\cite{SistlaClarke85}, as $\ltl$ is a fragment of $\pldl$.

The following is a polynomial space algorithm, which is correct due to Lemma~\ref{lemma_pumppath}: construct $\aut \times \sys$ on-the-fly and check whether it contains a pumpable fair path. Since the search for such a path can be implemented on-the-fly without having to construct the full product~\cite{KupfermanPitermanVardi09}, it can be implemented using polynomial space.
\end{proof}

To conclude, we prove the dual of Corollary~\ref{cor_mcub} for $\pldlbox$ formulas, which will be useful when we consider the model checking optimization problem. 

\begin{lemma}
	\label{lem_mcubbox}
Let $\varphi$ be a $\pldlbox$ formula and let $\sys$ be a transition system. There is a variable valuation~$\alpha^*$ that is bounded exponentially in $\size{\varphi}$ and linearly in $\size{\sys}$ such that if $\sys$ satisfies $\varphi$ with respect to $\alpha^*$, then $\sys$ satisfies $\varphi$ with respect to every valuation. 
\end{lemma}

\begin{proof}
We begin by defining $\alpha^*$: let $\aut$ be a Büchi automaton recognizing the models of $c(\neg \varphi)$, which is of exponential size in $\size{\varphi}$. Define $k^* = 4 \cdot \size{\aut} \cdot \size{\sys}+2$ and let $\alpha^*$ be the variable valuation mapping every variable to $k^*$. Now, we consider the contrapositive and show: if there is an $\alpha$ such that $\sys$ does not satisfy $\varphi$ with respect to $\alpha$, then $\sys$ does not satisfy $\varphi$ with respect to $\alpha^*$. 

Thus, assume there is an $\alpha$ and a path~$\pi$ such that $(\trace(\pi), \alpha) \models \neg \varphi$. Note that $\neg \varphi$ is a $\pldldiamond$-formula. Due to monotonicity, we can assume w.l.o.g.\ that $\alpha$ maps all variables to the same value, call it $k$.

We denote by $\trace(\pi)'$ the unique $k$-bounded and $k$-spaced $p$-coloring of $\trace(\pi)$ that starts with $p$ not holding true in the first position. Applying Item~\ref{lemma_alternatingcolor_pldltoldl} of Lemma~\ref{lemma_altcolor} shows that $\trace(\pi)'$ satisfies $c(\neg \varphi)$. Fix some accepting run of $\aut$ on $\trace(\pi)'$ and consider an arbitrary block of $\trace(\pi)'$: if the run does not visit an accepting state during the block, we remove infixes of the block and the run where the run reaches the same state before and after the infix and where the state of $\sys$ at the beginning and the end of the infix are the same, until the block has length at most $\size{\aut} \cdot \size{\sys}$. 

On the other hand, assume the run visits at least one accepting state during the block. Fix one such position. Then, we can remove infixes as above between the beginning of the block and the position before the accepting state is visited and between the position after the accepting state is reached and before the end of the block. What remains is a block of length at most $2\cdot\size{\aut}\cdot\size{\sys}+1$, at most $\size{\aut}\cdot \size{\sys}$ many positions before the designated position, this position itself, and at most $\size{\aut}\cdot \size{\sys}$ many after the designated position. 

Thus, we have constructed a $2 \cdot\size{\aut}\cdot \size{\sys}+1$-bounded coloring~$\trace(\hat{\pi})'$ of a trace~$\trace(\hat{\pi})$ for some path~$\hat{\pi}$ of $\sys$, as well as an accepting run of $\aut$ on $\trace(\hat{\pi})'$. Hence, $\trace(\hat{\pi})'$ is a model of $c(\neg \varphi)$ and applying Item~\ref{lemma_alternatingcolor_ldltopldl} of Lemma~\ref{lemma_altcolor} shows that $\trace(\hat{\pi})$ is a model of $\neg \varphi$ with respect to the variable valuation mapping every variable to $2\cdot(2\cdot\size{\aut}\cdot \size{\sys}+1)=k ^*$. Therefore, $\sys$ does not satisfy $\varphi$ with respect to $\alpha^*$.
\end{proof}

\section{Assume-guarantee Model Checking}
\label{sec_agmc}
After having solved the $\pldl$ model checking problem, we turn our attention to the assume-guarantee model checking problem. An instance of this problem consists of a transition system~$\sys$ and two specifications, an assumption~$\phiass$ and a guarantee~$\phigua$. Intuitively, whenever the assumption~$\phiass$ is satisfied, then also the guarantee~$\phigua$ should be satisfied.

More formally, given two transition systems~$\sys = (S, s_0, E, \ell)$ and $\sys' = (S', s_0', E', \ell')$ with $\ell(s_0) = \ell'(s_0')$, we define their parallel composition
\[\sys \parcomp \sys' =(S'', s_0'', E'', \ell'')\] where
\begin{itemize}
	\item $S'' = \set{ (s,s') \in S\times S' \mid \ell(s) = \ell'(s') }$,
	\item $s_0'' = (s_0, s_0')$,
	\item $((s, s'),(t,t')) \in E''$ if and only if $(s,t) \in E$ and $(s',t') \in E'$, and 
	\item 	$\ell''(s,s') = \ell(s) = \ell'(s')$. 
\end{itemize}  
Note that parallel composition as defined here amounts to taking the intersection of the trace languages of $\sys$ and $\sys'$. In particular, we have the following property.

\begin{remark}
\label{rem_parcomp}
Let $(v_0, v_0') (v_1, v_1') (v_2, v_2') \cdots $ be a path through a parallel composition~$\sys\parcomp\sys'$. Then, $v_0 v_1 v_2 \cdots$ is a path through $\sys$ that has the same trace as $(v_0, v_0') (v_1, v_1') (v_2, v_2') \cdots $.
\end{remark}

An assume-guarantee specification~$(\phiass, \phigua)$ consists of two $\pldl$ formulas, an assumption~$\phiass$ and a guarantee~$\phigua$. We say that a finite transition system~$\sys$ satisfies the specification, denoted by $\agmcdflt$, if for every countably infinite\footnote{This is the only place where we allow infinite transition systems (see the discussion below the proof of Lemma~\ref{lemma_agmc_charac}).}
transition system~$\sys'$, if $\sys \parcomp \sys'$ is a model of $\phiass$ with respect to some $\alpha$, then $\sys \parcomp \sys'$ is also a  model of $\phigua$ with respect to some $\beta$~\cite{Pnueli85}. For $\ltl$ specifications, this boils down to model checking the implication~$\phiass \rightarrow \phigua$, but the problem is more complex in the presence of parameterized operators, as already noticed by Kupferman et al.\ in the case of $\prompt$~\cite{KupfermanPitermanVardi09}. This is due to the fact that the variable valuation $\beta$ in the problem statement above may depend on $S'$.  In the following, we extend Kupferman et al.'s algorithm for the $\prompt$ assume-guarantee model checking problem to the $\pldl$ one. 

The main theorem of this section reads as follows.

\begin{theorem}
\label{thm_agmc}
The $\pldl$ assume-guarantee model checking problem is $\pspace$-complete.
\end{theorem}

To begin with, we show that we can refute such an assume-guarantee specification using a single trace per valuation~$\beta$, just like in the model checking problem where we are looking for a single counterexample. However, as we consider the satisfaction of two formulas, we have to deal with two variable valuations. 

\begin{lemma}
\label{lemma_agmc_charac}
Let $\sys$ be a transition system and let $(\phiass, \phigua)$ be a pair of $\pldl$ formulas. Then, $\agmcdflt$ does not hold if and only if there is a variable valuation~$\alpha$ such that for every variable valuation~$\beta$ there is a path~$\pi_\beta$ through~$\sys$ with $(\trace(\pi_\beta), \alpha) \models \phiass$, but $(\trace(\pi_\beta), \beta) \not\models \phigua$. 
\end{lemma}

\begin{proof}
Let $\agmcdflt$ not hold, i.e., there is a transition system~$\sys'$ such that $\sys \parcomp \sys'$ is a model of $\phiass$ with respect to some fixed $\alpha$, but $\sys \parcomp \sys'$ is not a  model of $\phigua$ with respect to any $\beta$. Thus, for every $\beta$, we find an initial path~$\pi'_\beta$ through $\sys \parcomp \sys'$ with $(\trace(\pi'_\beta), \beta) \not \models \phigua$. Furthermore, $\trace(\pi'_\beta)$ satisfies~$\phiass$ with respect to $\alpha$, as does every trace of $\sys \parcomp \sys'$. To conclude this direction, we apply Remark~\ref{rem_parcomp} to show that there exists a path $\pi_\beta$ over $\sys$ such that  $\trace(\pi_\beta) = \trace(\pi'_\beta)$ is also a trace of $\sys$. 

Now, assume there is a variable valuation~$\alpha$ such that for every variable valuation~$\beta$ there is a path~$\pi_{\beta}$ through~$\sys$ with $(\trace(\pi_\beta), \alpha) \models \phiass$, but $(\trace(\pi_\beta), \beta) \not\models \phigua$. Let $\sys'$ be a possibly infinite transition system whose traces are exactly the traces of the paths~$\pi_\beta$. By construction, the set of traces of $\sys \parcomp \sys'$ is equal to the set of traces of $\sys'$. Furthermore, every trace of $\sys \parcomp \sys'$ satisfies $\phiass$ with respect to $\alpha$. However, for every $\beta$ the trace~$\trace(\pi_\beta)$ of $\sys \parcomp \sys'$ does not satisfy $\phigua$ with respect to $\beta$. Hence, there is no $\beta$ such that $\sys \parcomp \sys'$ satisfies $\phigua$ with respect to $\beta$, i.e., $\sys'$ witnesses that $\agmcdflt$ does not hold.
\end{proof}

The $\prompt$ assume-guarantee model-checking problem as introduced in~\cite{KupfermanPitermanVardi09} only considers the product of the given transition system~$\sys$ with finite transition systems~$\sys'$. However, in this setting, one can easily construct counterexamples to the analogue of Lemma~\ref{lemma_agmc_charac}. Indeed the transition system~$\sys'$ we construct while proving the second implication above is necessarily infinite for the counterexamples. If one allows infinite systems~$\sys'$, then the analogue is still correct, using the same proof as above. The decidability of assume-guarantee model checking restricted to finite systems~$\sys'$ is an open problem.

Next, we observe that we again can restrict ourselves to considering  $\pldldiamond$ formulas, both as the assumption and as the guarantee. This follows from Lemma~\ref{lemma_removeboxes} and the fact that the variable valuations are quantified existentially in the problem statement. 

As we have to deal with two variable valuations, we have to extend the alternating-color technique to two colors, one color~$p$ for $\alpha$ and one color~$q$ for $\beta$. We say that $w' \in (\pow{P \cup \set{p,q}})^\omega$ is a coloring of $w \in (\pow{P})^\omega$, if $w_n' \cap P = w_n$ for every $n$. Furthermore, the notions of $p$-changepoints, $p$-blocks, and the analogues for $q$ are defined as expected (cf.\ Subsection~\ref{subsec_altcolor}). Consequently, the notions of $k$-boundedness and $k$-spacedness have to explicitly refer to the color under consideration. Lemma~\ref{lemma_altcolor} still holds for each color separately.

The following proof extends the one for the model checking problem using colored Büchi graphs. To this end, we have to adapt the definition of such a graph to two colors.  Formally, a colored Büchi graph of degree two is a tuple~$(V, E, v_0, \ell, F_0, F_1)$ where $(V, E)$ is a finite directed graph, $v_0 \in V$ is the initial vertex, $\ell \colon V \rightarrow \pow{\set{p,q}}$ is a vertex labeling by $p$ and $q$, and $F_0, F_1 \subseteq V$ are two sets of accepting vertices.

 A path $v_0 v_1 v_2 \cdots$ through $G$ is pumpable, if every $q$-block contains a vertex repetition such that there is a $p$-changepoint in between these vertices. More formally, we require the following condition to be satisfied: if $i$ and $i'$ are two adjacent $q$-changepoints, then there exist $j,j',j''$ with $i  \le j < j' < j'' < i'$ such that $v_j = v_{j''}$ and $\ell(v_j)$ and $\ell(v_{j'})$ differ in their $p$-label. Furthermore, the path is fair, if both $F_0$ and $F_1$ are visited infinitely often. 
 
The pumpable non-emptiness problem for $G$ asks whether there exists a pumpable fair path that starts in the initial vertex. 

\begin{theorem}[\cite{KupfermanPitermanVardi09}]
The pumpable non-emptiness problem for colored Büchi graphs of degree two is $\nlogspace$-complete.
\end{theorem}

Next, we show how to reduce the $\pldl$ assume-guarantee model checking problem to the pumpable non-emptiness problem for colored Büchi graphs of degree two. Fix an instance~$\agmcdflt$ of the problem with $\sys = (S, s_0, E, \ell)$ and two $\pldldiamond$ formulas $\phiass$ and $\phigua$. 

Now, let $\aut_A = (Q, \pow{P \cup \set{p,q}}, q_0, \delta, F)$ be a Büchi automaton recognizing the models of $\chi_{\infty p} \wedge\chi_{\infty \neg p} \wedge \rel{\phiass}$, and let $\aut_G = (Q', \pow{P \cup \set{p,q}}, q_0', \delta', F')$ be a Büchi automaton recognizing the models of $\chi_{\infty q} \wedge \chi_{\infty \neg q} \wedge \neg\rel{\phigua}$. Note that we need to slightly adapt the construction of $\aut_G$, as we interpret the changepoint-bounded operators in $\phigua$ w.r.t.\ color changes of $q$, not $p$. Hence, instead of using the automaton~$\aut_{cp}$ as depicted in Figure~\ref{fig_atomicaut}(c) with transition labels~$p$ and $\neg p$, we use the one with labels~$q$ and $\neg q$ to construct $\aut_G$.  

Next, we define the colored Büchi graph of degree two 
\[\aut_A \times \aut_G \times \sys = (Q \times Q' \times S \times \pow{\set{p,q}}, E', (q_0, q_0', s_0, \emptyset), \ell', F_0, F_1)\] where
\begin{itemize}
	\item $((q_1,q_2, s, C),(q_1', q_2', s', C')) \in E'$ if and only if $(s, s') \in E$, $q_1' \in \delta(q_1, \ell(s) \cup C)$, and $q_2' \in \delta'(q_2, \ell(s) \cup C )$,
	\item $\ell' (q_1,q_2, s, C) = C$,
	\item $F_0 = F \times Q' \times S \times \pow{\set{p,q}}$, and
	\item $F_1 = Q \times F' \times S \times \pow{\set{p,q}}$.
\end{itemize}

\begin{lemma}[cf.\ Lemma~6.2 of \cite{KupfermanPitermanVardi09}]
\label{lemma_pumpcharac}
Let $\agmcdflt$ and $\aut_A \times \aut_G \times \sys$ be defined as above. Then, $\agmcdflt$ does not hold if and only if $\aut_A \times \aut_G \times \sys$ has a pumpable fair path.
\end{lemma}

\begin{proof}
Recall that changepoint-bounded operators in $\phiass$ are evaluated with respect to the color~$p$ while the ones in $\phigua$ are evaluated with respect to $q$.

Let $\agmcdflt$ not hold. Then, due to Lemma~\ref{lemma_agmc_charac}, there is a variable valuation~$\alpha$ such that for every valuation~$\beta$ there is an initial path~$\pi_\beta$ of $\sys$ such that $(\trace(\pi_\beta), \alpha) \models \phiass$, but $(\trace(\pi_\beta), \beta) \not\models \phigua$. 	

Define $k_\alpha = \max_{x\in \var(\phiass)}\alpha(x)$, $k_{\beta^*} = 2\cdot \size{Q} \cdot \size{Q'} \cdot \size{S} \cdot  k_\alpha+1$, and let $\beta^*$ be such that $\beta^*(x) = 2k_{\beta^*}$ for every $x$. Finally,  let $w^* = \trace(\pi_{\beta^*})$ be the corresponding trace as above.

 Then, $(w^*, \alpha) \models \phiass$ and Item \ref{lemma_alternatingcolor_pldltoldl} of Lemma~\ref{lemma_altcolor} imply that every $k_{\alpha}$-bounded (with respect to $p$) coloring~${w^*}'$ of $w^*$ satisfies $  \chi_{\infty p} \wedge \chi_{\infty \neg p} \wedge \rel{\phiass}$.
Similarly, $(w^*, \beta^*) \not\models \phigua$ and Item~\ref{lemma_alternatingcolor_ldltopldl} of Lemma~\ref{lemma_altcolor} imply that every $k_{\beta^*}$-spaced (with respect to $q$) coloring~${w^*}'$ of $w^*$ does not satisfy $ \rel{\phigua}$. Hence, every such ${w^*}'$ satisfies $\chi_{\infty q} \wedge \chi_{\infty \neg q}  \wedge \neg \rel{\phigua}$.

Now, consider the unique coloring~${w^*}'$ of $w^*$ that is $k_\alpha$-bounded and $k_\alpha$-spaced with respect to $p$, $k_{\beta^*}$-bounded and $k_{\beta^*}$-spaced with respect to $q$, and begins with $p$ and $q$ not holding true. We have ${w^*}' \models \chi_{\infty p} \wedge \chi_{\infty \neg p} \wedge \rel{\phiass}$ and ${w^*}' \models \chi_{\infty q} \wedge \chi_{\infty \neg q} \wedge \neg \rel{\phigua}$. Hence, there are accepting runs~$q_0q_1q_2\cdots$ of $\aut_A$ and $q_0'q_1'q_2'\cdots$ of $\aut_G$ on ${w^*}'$. 

Consider the path
\[
(q_0, q_0', v_0, {w^*_0}'\cap \set{p,q})\,
(q_1, q_1', v_1, {w^*_1}'\cap \set{p,q})\,
(q_2, q_2', v_2, {w^*_2}'\cap \set{p,q})
\cdots
\]
through $\aut_A \times \aut_G \times \sys$. Here, ${w^*_n}'$ is the $n$-th letter of ${w^*}'$ and $v_0 v_1 v_2 \cdots$ is the path through $\sys$ inducing the trace~${w^*}'$.

The path is fair, as the runs are both accepting. Furthermore, it is pumpable, as the $p$-blocks are of size $k_\alpha$, but the $q$-blocks are of length~$k_{\beta^*} = 2\cdot \size{Q} \cdot \size{Q'} \cdot \size{S} \cdot  k_\alpha+1$ and there are only $2\cdot \size{Q} \cdot \size{Q'} \cdot \size{S}$ many vertices with (and without) color~$q$.

Now, we consider the converse: assume there is a pumpable fair path
\[
(q_0, q_0', v_0, C_0)\,
(q_1, q_1', v_1, C_1)\,
(q_2, q_2', v_2, C_2)
\cdots
\]
in $\aut_A \times \aut_G \times \sys$. W.l.o.g., we can assume the path to be ultimately periodic~\cite{KupfermanPitermanVardi09}. Hence, the maximal length of a $p$-block in this path, call it $k_\alpha$, is well-defined. Define $\alpha$ via $\alpha(x) = 2k_\alpha$ for every $x$, fix some arbitrary $\beta$, and let $k_\beta = \max_{x \in \var(\phigua)}\beta(x)$. 

Every $q$-block of the pumpable path contains a vertex repetition with a $p$-changepoint in between. Pumping each of these repetitions $k_\beta$ times yields a new path through $\aut_A \times \aut_G \times \sys$ and thereby also a path~$\pi_\beta$ through $\sys$ as well as accepting runs of $\aut_A$ and $\aut_G$ on a coloring~$w'$ of $\trace(\pi_\beta)$. Hence, $w' \models \chi_{\infty p} \wedge \chi_{\infty \neg p} \wedge \rel{\phiass}$ and $w' \models \chi_{\infty q} \wedge \chi_{\infty \neg q} \wedge \neg \rel{\phigua}$. 

By construction, $w'$ is $k_\alpha$-bounded and $k_\beta$-spaced. Thus, applying both directions of Lemma~\ref{lemma_altcolor} yields $(\trace(\pi_\beta), \alpha) \models \phiass$ and $(\trace(\pi_\beta), \beta) \not \models \phigua$. Hence, for every $\beta$ we have constructed a path with the desired properties. Thus, due to Lemma~\ref{lemma_agmc_charac}, $\agmcdflt$ does not hold.
\end{proof}

Now, we are able prove the main result of this section: $\pldl$ assume-guar\-antee model checking is as hard as $\ltl$ assume-guarantee model checking, i.e., $\pspace$-complete. 

\begin{proof}[Proof of Theorem~\ref{thm_agmc}]
Membership is obtained by solving the pumpable non-emptiness problem for the product~$\aut_A \times \aut_G \times \sys $, which can be done in polynomial space on-the-fly, as the product is of exponential size and the algorithm checking for pumpable non-emptiness runs in logarithmic space. 

For the lower bound we use a reduction from the $\ltl$ model checking problem, which is $\pspace$-complete: given a transition system~$\sys$ and an $\ltl$ formula~$\varphi$, we have $\sys \models \varphi$ if and only if $\agmc{\ttrue}{\sys}{\varphi}$.
\end{proof}

The solution to the assume-guarantee model checking problem also solves the implication problem for $\pldl$: given two $\pldl$ formulas $\varphi$ and $\psi$, decide whether for every, possibly countably infinite, transition system~$\sys$ the following holds: if $\sys$ satisfies $\varphi$ with respect to some $\alpha$, then $\sys$ satisfies $\psi$ with respect to some $\beta$. 

\begin{theorem}
The $\pldl$ implication problem is $\pspace$-complete.	
\end{theorem}

\begin{proof}
Hardness follows from hardness of the $\ltl$ satisfiability problem~\cite{SistlaClarke85}. 

To prove membership, we reduce the problem to the assume-guarantee model checking problem: let $\mathcal{U}$ be a universal transition system in the sense that it contains every trace over the propositions that appear in $\varphi$ and $\psi$. It is straightforward to show that the implication between $\varphi$ and $\psi$ is satisfied, if and only if $\agmc{\varphi}{\mathcal{U}}{\psi}$ is satisfied, as $\mathcal{U} \parcomp \sys$ has exactly the traces of $\sys$. The latter problem can be solved in $\pspace$, although  $\mathcal{U}$ is of exponential size, since it can be constructed on-the-fly. 
\end{proof}

\section{Realizability}
\label{sec_real}
In this section, we consider the realizability problem for $\pldl$. Throughout the section, we fix a partition~$(I, O)$ of the set of atomic propositions~$P$. An instance of the $\pldl$ realizability problem is given by a $\pldl$ formula~$\varphi$ (over $P$) and the problem is to decide whether Player~$O$ has a winning strategy in the following game, played in rounds~$n \in \nats$: in each round~$n$, Player~$I$ picks a subset $i_n \subseteq I$ and then Player~$O$ picks a subset~$o_n \subseteq O$. Player~$O$ wins the play with respect to a variable valuation~$\alpha$, if 
$((i_0 \cup o_0)(i_1 \cup o_1)(i_2 \cup o_2) \cdots, \alpha) \models \varphi$.

Formally, a strategy for Player~$O$ is a mapping~$\sigma\colon (\pow{I})^+ \rightarrow \pow{O}$ and a play~$\rho = i_0 o_0 i_1 o_1 i_2 o_2 \cdots $ is consistent with $\sigma$, if we have $o_n = \sigma(i_0 \cdots i_n)$ for every $n$. We call $(i_0 \cup o_0)(i_1 \cup o_1)(i_2 \cup o_2) \cdots$ the outcome of $\rho$, denoted by $\outcome(\rho)$. We say that a strategy~$\sigma$ for Player~$O$ is winning with respect to a variable valuation~$\alpha$, if we have $(\outcome(\rho), \alpha) \models \varphi$ for every play~$\rho$ that is consistent with $\sigma$. The $\pldl$ realizability problem asks for a given $\pldl$ formula~$\varphi$, whether Player~$O$ has a winning strategy with respect to some variable valuation, i.e., whether there is a single $\alpha$ such that every outcome satisfies $\varphi$ with respect to $\alpha$. If this is the case, then we say that $\sigma$ realizes $\varphi$ and thus that $\varphi$ is realizable (over $(I,O)$).

It is well-known that $\omega$-regular specifications, and thus all $\ldlt$ specifications, are realizable by finite-state transducers (if they are realizable at all)~\cite{BuechiLandweber69}. A transducer~$\trans = (Q, \Sigma, \Gamma, q_0, \delta, \tau )$ consists of a finite set~$Q$ of states, an input alphabet~$\Sigma$, an output alphabet~$\Gamma$, an initial state~$q_0$, a transition function~$\delta \colon Q \times \Sigma \rightarrow Q$, and an output function~$\tau \colon Q \rightarrow \Gamma$. The function~$f_{\trans} \colon \Sigma^* \rightarrow \Gamma$ implemented by $\trans$ is defined as $f_{\trans}(w) = \tau(\delta^*(w))$, where $\delta^*$ is defined as usual: $\delta^*(\epsilon) = q_0$ and $\delta^*(wv) = \delta(\delta^*(w),v)$. To implement a strategy by a transducer, we use $\Sigma = \pow{I}$ and $\Gamma = \pow{O}$. Then, we say that the strategy $\sigma = f_\trans$ is finite-state. The size of $\sigma$ is the number of states of $\trans$. The following proof is analogous to the one for $\prompt$~\cite{KupfermanPitermanVardi09}.

\begin{theorem}
\label{thm_real}
The $\pldl$ realizability problem is $\twoexp$-complete.
\end{theorem}

When proving membership in $\twoexp$, we restrict ourselves w.l.o.g.\ to $\pldldiamond$ formulas, as this special case is sufficient as shown in Lemma~\ref{lemma_removeboxes}. First, we use the alternating color technique to show that the $\pldldiamond$ realizability problem is reducible to the realizability problem for specifications in $\ldlt$. When considering the $\ldlt$ realizability problem, we add the fresh proposition~$p$ used to specify the coloring to $O$, i.e., Player~$O$ is in charge of determining the color of each position.

\begin{lemma}[cf.\ Lemma 3.1 of \cite{KupfermanPitermanVardi09}]
	\label{lemma_realred}
A $\pldldiamond$ formula~$\varphi$ is realizable over $(I,O)$ if and only if the $\ldlt$ formula~$c(\varphi)$ is realizable over $(I,O \cup \set{p})$.
\end{lemma}

\begin{proof}
Let $\varphi$ be realizable, i.e., there is a winning strategy~$\sigma\colon (\pow{I})^+ \rightarrow \pow{O} $ with respect to some~$\alpha$. Now, consider the strategy~$\sigma'\colon (\pow{I})^+ \rightarrow \pow{O \cup \set{p}}$ defined by
\[\sigma'(i_0 \cdots i_{n-1}) = 
\begin{cases}
\sigma(i_0 \cdots i_{n-1})					&\text{if $n \bmod 2k < k$,}\\
\sigma(i_0 \cdots i_{n-1}) \cup \set{p}		&\text{otherwise,}
\end{cases}\]
where $k = \max_{x \in \vardiamond(\varphi)} \alpha(x)$. We show that $\sigma'$ realizes $c(\varphi)$. To this end, let $\rho' = i_0 o_0 i_1 o_1 i_2 o_2 \cdots$ be a play that is consistent with $\sigma'$. Then, 
$ \rho = i_0 (o_0 \setminus \set{p}) i_1 (o_1 \setminus \set{p}) i_2 (o_2 \setminus \set{p}) \cdots$ is by construction consistent with $\sigma$, i.e., $(\outcome(\rho), \alpha)\models \varphi$. As $\outcome(\rho')$ is a $k$-spaced $p$-coloring of $\outcome(\rho)$, we deduce $\outcome(\rho') \models c(\varphi)$ by applying Item~\ref{lemma_alternatingcolor_pldltoldl} of Lemma~\ref{lemma_altcolor}. Hence, $\sigma'$ realizes $c(\varphi)$.

Now, assume $c(\varphi)$ is realized by $\sigma'\colon (\pow{I})^+ \rightarrow \pow{O \cup \set{p}}$, which we can assume to be finite-state, say it is implemented by $\trans$ with $n$ states. We first show that every outcome that is consistent with $\sigma'$ is $n+1$-bounded. Such an outcome satisfies $c(\varphi)$ and has therefore infinitely many changepoints. Now, assume it has a block of length strictly greater than $n+1$, e.g., between changepoints at positions~$i$ and $j$. Let $q_0 q_1 q_2 \cdots$ be the states reached during the run of $\trans$ on the projection of $\rho$ to $\pow{I}$. Then, there are two positions~$i'$ and $j'$ satisfying $i \le i' < j' < j$ in the block such that $q_{i'} = q_{j'}$. Hence, $q_0 \cdots q_{i'-1}(q_{i'} \cdots q_{j'-1})^\omega$ is also a run of $\trans$. However, the output generated by this run has only finitely many changepoints, since the output at the states~$q_{i'}, \ldots, q_{j'-1}$ coincides when restricted to $\set{p}$. This contradicts the fact that $\trans$ implements a winning strategy, which requires in particular that every output has infinitely many changepoints. Hence, $\rho$ is $(n+1)$-bounded.

Let $\sigma \colon (\pow{I})^+ \rightarrow \pow{O}$ be defined as $\sigma(i_0 \cdots i_{n-1}) = \sigma'(i_0 \cdots i_{n-1}) \cap O$.
By definition, for every play~$\rho$ consistent with $\sigma$, there is an $(n+1)$-bounded $p$-coloring of $\outcome(\rho)$ that is the outcome of a play that is consistent with $\sigma'$. Hence, applying Item~\ref{lemma_alternatingcolor_ldltopldl} of Lemma~\ref{lemma_altcolor} yields $(\outcome(\rho), \beta) \models \varphi$, where $\beta(x) = 2n+2$ for every $x$. Hence, $\sigma$ realizes $\varphi$ with respect to $\beta$. Note that $\sigma$ is also finite-state and of the same size as $\sigma'$.
\end{proof}

Now, we are able to prove the main result of this section.

\begin{proof}[Proof of Theorem~\ref{thm_real}]
$\twoexp$-hardness of the $\pldl$ realizability problem follows immediately from the $\twoexp$-hardness of the $\ltl$ realizability problem~\cite{PnueliRosner89a}, as $\ltl$ is a fragment of $\pldl$.

Now, consider membership and recall that we have argued that it is sufficient to consider $\pldldiamond$ formulas. Thus, let $\varphi$ be a $\pldldiamond$ formula. By Lemma~\ref{lemma_realred}, we know that it is sufficient to consider the realizability of $c(\varphi)$. Let $\aut = (Q, \pow{I \cup O \cup \set{p}}, q_0, \delta, \col)$ be a deterministic parity automaton recognizing the models of $c(\varphi)$. We turn $\aut$ into a parity game~$\game$ such that Player~$1$ wins $\game$ from some dedicated initial vertex if and only if $c(\varphi)$ is realizable.
To this end, we define the arena $\arena = (V, V_0, V_1, E)$ with 
\begin{itemize}
	\item $V = Q \cup (Q \times \pow{I})$,
	\item $V_0 = Q$, 
	\item $V_1 = Q \times \pow{I}$, and
	\item $E = \set{(q, (q,i)) \mid i \subseteq I} \cup \set{((q, i), \delta(q, i \cup o)) \mid o \subseteq O \cup \set{p}}$,
i.e., Player~$0$ picks a subset $i \subseteq I$ and Player~$O$ picks a subset~$o \subseteq O \cup \set{p}$, which in turn triggers the (deterministic) update of the state stored in the vertices. 
\end{itemize}
Finally, we define the coloring~$\col_\arena$ of the arena via $\col_\arena(q) = \col_\arena(q,i) =\col(q)$. 

It is straightforward to show that Player~$O$ has a winning strategy from $q_0$ in the parity game~$(\arena, \col_\arena)$ if and only if $c(\varphi)$ (and thus $\varphi$) is realizable. Furthermore, if Player~$1$ has a winning strategy, then $\arena$ can be turned into a transducer implementing a strategy that realizes $c(\varphi)$ using $V$ as set of states. Note that $\card{V}$ is doubly-exponential in $\card{\varphi}$, if we assume that $I$ and $O$ are restricted to propositions appearing in $\varphi$. As the parity game is of doubly-exponential size and has exponentially many colors, we can solve it in doubly-exponential time in the size of $\varphi$. \end{proof}

Also, we obtain a doubly-exponential upper bound on a valuation that allows to realize a given formula. A matching lower bound already holds for $\pltl$~\cite{Zimmermann13}.

\begin{corollary}
	\label{cor_realub}
If a $\pldldiamond$ formula~$\varphi$ is realizable with respect to some $\alpha$, then it is realizable with respect to an $\alpha$ that is bounded doubly-exponentially in $\card{\varphi}$. 
\end{corollary}

\begin{proof}
If $\varphi$ is realizable, then so is $c(\varphi)$. Using the construction proving the right-to-left implication of Lemma~\ref{lemma_realred}, we obtain that $\varphi$ is realizable with respect to some $\alpha$ that is bounded by $2n+2$, where $n$ is the size of a transducer implementing the strategy that realizes $c(\varphi)$. We have seen in the proof of Theorem~\ref{thm_real} that the size of such a transducer is at most doubly-exponential in $\card{c(\varphi)}$, which is only linearly larger than $\card{\varphi}$. The result follows. 
\end{proof}

\section{Optimal Variable Valuations for Model Checking and Realizability}
\label{sec_optimal}
In this section, we turn the model checking and the realizability problem into optimization problems, e.g., the model checking optimization problem asks for the \emph{optimal} variable valuation such that a given system satisfies the specification with respect to this valuation. Similarly, the realizability optimization problem asks for an \emph{optimal} variable valuation such that $\varphi$ is realizable with respect to this valuation. Furthermore, we are interested in computing a winning strategy for Player~$O$ witnessing realizability with respect to an optimal valuation.
The definition of optimality depends on the type of formula under consideration: for $\pldldiamond$ formulas, we want to minimize the waiting times while for $\pldlbox$ formulas, we want to maximize satisfaction times. For formulas having both types of parameterized operators, the optimization problems are undefined.

In Subsection~\ref{subsec_optmc}, we show how to solve the model checking optimization problem in polynomial space. Then, in Subsection~\ref{subsec_optrealizability}, we explain how to adapt the approach to solve the realizability optimization problem in triply-exponential time. Thus, the model checking optimization problem is in polynomial space, just as the decision problem, but there is an exponential gap between the realizability optimization problem and its decision variant. Note that this gap already exists for $\pltl$~\cite{Zimmermann13}.

Both our results rely on the existence of automata of a certain size that recognize the models of a given $\pldl$ formula with respect to a fixed variable valuation. On the one hand, it suffices to translate formulas with a single variable; on the other hand, due to some technicalities, we have to consider formulas that might additionally contain changepoint-bounded operators. The semantics of such formulas are defined as expected.  

\begin{theorem}
\label{thm_smallautomata}
Let $\varphi$ be a $\pldl$ formula with $\var(\varphi) = \set{z}$ possibly having change\-point-bounded operators and let $\alpha$ be a variable valuation. Then, there exists a natural number $n \in (3 \cdot(\alpha(z)+1))^{\mathcal{O}(\size{\varphi})}$ and there exist
\begin{enumerate}
	
	\item\label{thm_smallautomata_nondet}
	 a non-deterministic Büchi automaton of size $n$ and

	\item\label{thm_smallautomata_det} a deterministic parity automaton of size $(n!)^2$ with $2n$ many colors
		
\end{enumerate}
that recognize the language~$L(\varphi, \alpha) = \set{ w \in (2^{P'})^\omega \mid (w, \alpha) \models \varphi }$, which are both effectively constructible.

\end{theorem}

The existence of such automata is proven in Subsection~\ref{subsec_smallautomata} by adapting the Breakpoint construction of Miyano and Hayashi~\cite{MiyanoH84}.

\subsection{The Model Checking Optimization Problem}
\label{subsec_optmc}
In this subsection, we prove that the model checking optimization problem can be solved in polynomial space. As already mentioned above, we only consider $\pldldiamond$ and $\pldlbox$ formulas.
For $\pldldiamond$ formulas, optimal variable valuations are as small as possible. 
To abstract a variable valuation to a single value, we can either take the minimal element among the variables, i.e. the shortest waiting time, or the maximal element among the variables, i.e. the longest waiting time. Both options will provide a total ordering among variable valuations.
For $\pldlbox$ formulas, optimal variable valuations are as large as possible. For abstraction purposes, we may again either take the maximal element, i.e. the longest guarantee, or the minimal element among the variable, i.e. the shortest guarantee. Again, this results in a total order.

\begin{theorem}
	\label{thm_optmc}
Let $\varphi_\Diamond$ be a $\pldldiamond$ formula, let $\varphi_\Box$ be a $\pldlbox$ formula, and let $\sys$ be a transition system. The following values are computable in polynomial space:
\begin{enumerate}

	\item $\min_{\set{\alpha \mid \sys \text{ satisfies }\varphi_\Diamond \text{ w.r.t.\ } \alpha}} \min_{x \in \var(\varphi_\Diamond)} \alpha(x)$.
	
	\item $\min_{\set{\alpha \mid \sys \text{ satisfies } \varphi_\Diamond \text{ w.r.t.\ } \alpha}} \max_{x \in \var(\varphi_\Diamond)} \alpha(x)$.
	
	\item $\max_{\set{\alpha \mid \sys \text{ satisfies } \varphi_\Box \text{ w.r.t.\ } \alpha}} \max_{y \in \var(\varphi_\Box)} \alpha(y)$.

	\item $\max_{\set{\alpha \mid \sys \text{ satisfies } \varphi_\Box \text{ w.r.t.\ } \alpha}} \min_{y \in \var(\varphi_\Box)} \alpha(y)$.
	
\end{enumerate}
\end{theorem}

Note that all other combinations are trivial due to the monotonicity properties of $\pldl$. Furthermore, we can restrict our attention to formulas with at least one variable, as the optimization problem is trivial otherwise. 

As a first step, we show that we can reduce all problems to ones with exactly one variable, but possibly with changepoint-bounded operators.

\begin{enumerate}
	\item Fix some $x \in \var(\varphi_\Diamond)$ and apply the rewriting introduced for the alter\-nating-color technique to every variable but $x$ to obtain the formula~$\varphi_x$, which has changepoint-bounded diamond-operators as well as diamond-operators parameterized by $x$.  Applying both directions of Lemma~\ref{lemma_altcolor} (which also holds if we do not replace all parameterized operators) yields
\begin{align*}
&\min\nolimits_{\set{\alpha \mid \sys \text{ satisfies } \varphi_\Diamond \text{ w.r.t.\ } \alpha}} \min\nolimits_{x \in \var(\varphi_\Diamond)} \alpha(x) =\\
 &\min\nolimits_{x \in \var(\varphi_\Diamond)} \min\nolimits_{\set{\alpha \mid  \sys \text{ satisfies } \varphi_x \text{ w.r.t.\ } \alpha}} \alpha(x).
\end{align*}
Thus, we have reduced the problem to $\size{\var(\varphi_\Diamond)}$ many optimization problems for formulas~$\varphi_x$ with a single variable.
	
\item Rename every variable in $\varphi_\Diamond$ to $z$ and call the resulting formula $\varphi_\Diamond'$. Due to monotonicity, minimizing the maximal parameter value for $\varphi_\Diamond$ yields the same value as minimizing the value of $z$ for $\varphi_\Diamond'$.	

\item Fix some $y \in \var(\varphi_\Box)$ and denote by $\varphi_y$ the formula obtained from $\varphi_\Box$ by replacing every subformula~$\bboxle{r}{y'}\psi$ with $y' \neq y$ by $\bbox{\hat{r}}\psi$, where $\hat{r}$ is defined as in the proof of Lemma~\ref{lemma_removeboxes}.
Intuitively, this sets the value for every $y' \neq y $ to zero. Due to monotonicity, we have
\begin{align*}
&\max\nolimits_{\set{\alpha \mid \sys \text{ satisfies } \varphi_\Box \text{ w.r.t.\ } \alpha}} \max\nolimits_{y \in \var(\varphi_\Box)} \alpha(y) =\\ &\max\nolimits_{y \in \var(\varphi_\Box)} \max\nolimits_{\set{\alpha \mid \sys \text{ satisfies } \varphi_y \text{ w.r.t.\ } \alpha}} \alpha(y),
\end{align*}
i.e., we have reduced the problem to $\size{\var(\varphi_\Box)}$ many optimization problems for formulas~$\varphi_y$ with a single variable.

\item Rename every variable in $\varphi_\Box$ to $z$ and call the resulting formula $\varphi_\Box'$. Due to monotonicity, maximizing the minimal parameter value for $\varphi_\Box$ yields the same value as maximizing the value of $z$ for $\varphi_\Box'$.	

\end{enumerate}

First, we consider the minimization problem for a formula~$\varphi_\Diamond$ with a single variable~$x \in \vardiamond(\varphi_\Diamond)$ and possibly with 
changepoint-bounded operators. From Corollary~\ref{cor_mcub}, which can easily be shown to hold for such formulas, too, we obtain an upper bound (that is exponential in $\size{\varphi_\Diamond}$ and linear in $\size{\sys}$) on the value
\[ \min\nolimits_{\set{\alpha \mid \sys \text{ satisfies } \varphi_\Diamond \text{ w.r.t.\ } \alpha}} \alpha(x). \]
Dually, for a formula~$\varphi_\Box$ with a single variable~$y \in \varbox(\varphi_\Box)$ and possibly with changepoint-bounded operators, Lemma~\ref{lem_mcubbox}, which holds for such formulas, too, yields a bound~$k_{\max}$ (that is exponential in $\size{\varphi_\Box}$ and linear in $\size{\sys}$) such that either
\[ \max\nolimits_{\set{\alpha \mid \sys \text{ satisfies } \varphi_\Box \text{ w.r.t.\ } \alpha}} \alpha(y) \le k_{\max} \]
or the maximum is equal to $\infty$. Thus, in both cases, we have an exponential search space for the optimal value.

Therefore, binary search yields the optimal value, if we can solve each query \myquot{does $\sys$ satisfy $\varphi$ with respect to $\alpha$} in polynomial space, provided $\alpha$ is exponential in $\size{\varphi}$ and linear in $\size{\sys}$. To this end, we use the non-deterministic Büchi automaton~$\aut$ recognizing $L(\neg \varphi, \alpha)$ as given by Item~\ref{thm_smallautomata_nondet} of Theorem~\ref{thm_smallautomata}, whose size is exponential in $\size{\varphi}$ and linear in $\size{\sys}$. Model checking $\sys$ against $\aut$ answers the query and is possible in polynomial space by executing the emptiness test on-the-fly without constructing $\aut$ completely~\cite{VardiWolper94}.

\subsection{The Realizability Optimization Problem}
\label{subsec_optrealizability}
In this subsection, we show how to adapt the reasoning of the model checking case to give an algorithm for the realizability optimization problem with triply-exponential running time.

\begin{theorem}	
	\label{thm_optreal}
Let $\varphi_\Diamond$ be a $\pldldiamond$ formula and let $\varphi_\Box$ be a $\pldlbox$ formula. The following values (and winning strategies witnessing them) can be computed in triply-exponential time:
\begin{enumerate}

	\item $\min_{\set{\alpha \mid \varphi_\Diamond \text{ realizable w.r.t.\ } \alpha}} \min_{x \in \var(\varphi_\Diamond)} \alpha(x)$.
	
	\item $\min_{\set{\alpha \mid \varphi_\Diamond \text{ realizable w.r.t.\ } \alpha}} \max_{x \in \var(\varphi_\Diamond)} \alpha(x)$.
	
	\item $\max_{\set{\alpha \mid \varphi_\Box \text{ realizable w.r.t.\ } \alpha}} \max_{y \in \var(\varphi_\Box)} \alpha(y)$.

	\item $\max_{\set{\alpha \mid \varphi_\Box \text{ realizable w.r.t.\ } \alpha}} \min_{y \in \var(\varphi_\Box)} \alpha(y)$.
	
\end{enumerate}
\end{theorem}

The reductions to optimization problems for formulas with a single variable remain valid in the realizability case. However, instead of proving bounds on the search space for both the $\pldldiamond$ case and the $\pldlbox$ case, we rely on Corollary~\ref{cor_realub}, which proves an upper bound for the former case, and on duality: given a $\pldl$ formula~$\varphi$ over $P = I \cup O$ and its negation~$\neg \varphi$ as defined in Lemma~\ref{lemma_pldlnegation}, define $\dual{\varphi}$ to be the formula obtained from $\neg\varphi$ by replacing each atomic proposition~$p \in I$ by $\ddiamond{\ttrue}p$ and each negated proposition~$\neg p$ with $p \in I$ by $\ddiamond{\ttrue}\neg p$. Here, $\ddiamond{\ttrue}\!$ can be understood as the $\pldl$-equivalent of $\ltl$'s next-operator. The realizability problems for $\varphi$ and $\dual{\varphi}$ are dual, i.e., we have swapped the roles of the players and negated the specification (and used the next-operator to account for the fact that Player~$I$ is always the first to move). The following lemma formalizes this fact and relies on determinacy of parity games~\cite{EmersonJutla91,Mostowski91}, to which the realizability problem is reduced to, as shown in Section~\ref{sec_real}.

\begin{lemma}
	\label{lemma_dualrealability}
Let $\varphi$ be a $\pldl$ formula and let $\alpha$ be a variable valuation. Then, $\varphi$ is not realizable over $(I, O)$ with respect to $\alpha$ if and only if $\dual{\varphi}$ is realizable over $(O, I)$ with respect to $\alpha$.
\end{lemma}

Thus, applying Lemma~\ref{lemma_dualrealability} and monotonicity in the case of a $\pldlbox$ formula~$\varphi_\Box$ with a single variable~$y$ yields 
\[
\max\nolimits_{\set{\alpha \mid \varphi_\Box \text{ realizable w.r.t.\ } \alpha}} \alpha(y) = \min\nolimits_{\set{\alpha \mid \dual{\varphi_\Box} \text{ realizable w.r.t.\ } \alpha}} \alpha(y) -1, \]
i.e., to solve the $\pldlbox$ optimization problem for $\varphi_\Box$ we just have to solve the problem for $\dual{\varphi_\Box}$ and subtract one.

Thus, it remains to consider a minimization problem for a formula~$\varphi_\Diamond$ with a single variable~$x \in \vardiamond(\varphi_\Diamond)$ and possibly with 
changepoint-bounded operators. From Corollary~\ref{cor_realub}, which holds for such formulas, too, we obtain a doubly-exponential (in $\size{\varphi_\Diamond}$) upper bound on
$ \min\nolimits_{\set{\alpha \mid \varphi_\Diamond \text{ realizable w.r.t.\ } \alpha}} \alpha(x)$.

Thus, we have a doubly-exponential search space for the optimal variable valuation. Recall that Item~\ref{thm_smallautomata_det} of Theorem~\ref{thm_smallautomata} gives us a deterministic parity automaton of triply-exponential size and with exponentially many colors (both in $\size{\varphi_\Diamond}$) recognizing $L(\varphi_\Diamond, \alpha)$, as $\alpha(x)$ is bounded doubly-exponentially. This allows us to construct a parity game of triply-exponential size with exponentially many colors that is won by Player~$O$ if and only if $\varphi_\Diamond$ is realizable with respect to $\alpha$. The construction is similar to the one described in the proof of Theorem~\ref{thm_real}. Such a parity game can be solved in triply-exponential time. Thus, to solve the optimization problem, we perform binary search through the doubly-exponential search space where each query can be answered in triply-exponential time by solving a parity game. Thus, the overall running time is indeed triply-exponential.

Furthermore, as already described in the aforementioned proof, a winning strategy for the parity game can be turned into a transducer witnessing realizability of $\varphi_\Diamond$. Finally, it is straightforward to show how to turn this transducer into one for the original specifications with potentially several variables. This is trivial for the cases not requiring an application of the alternating-color technique and requires the transformation described in the proof of Lemma~\ref{lemma_realred} for the other cases. This finishes the proof of Theorem~\ref{thm_optreal}, save for the construction of a deterministic parity automaton with the desired properties.

\subsection{\emph{Small} Automata for PLDL}
\label{subsec_smallautomata}
Fix a formula~$\varphi$ with a single variable~$z \in \var(\varphi)$ possibly having change\-point-bounded operators and a variable valuation~$\alpha$. We show how to adapt the Breakpoint construction of Miyano and Hayashi~\cite{MiyanoH84} to construct a non-deterministic Büchi automaton of size~$(3 \cdot(\alpha(z)+1))^{\mathcal{O}(\size{\varphi})}$ to prove Item~\ref{thm_smallautomata_nondet} of Theorem~\ref{thm_smallautomata}. The deterministic automaton for Item~\ref{thm_smallautomata_det} of Theorem~\ref{thm_smallautomata} can then be obtained by applying Schewe's determinization construction~\cite{Schewe09}, which determinizes a Büchi automaton with $n$ states into a parity automaton with $(n!)^2$ states and $2n$ colors.

Recall that $\varphi$ has a single variable. In the following, we assume that it parameterizes diamond-operators, the case of box-operators is dual and discussed below. Thus, call the variable~$x$ and let $\ddiamondle{r_1}{x}\psi_1, \ldots, \ddiamondle{r_k}{x}\psi_k \in \cl(\varphi)$ be the parameterized subformulas of $\varphi$. Furthermore, let $\varphi'$ be the $\ldlt$ formula obtained by removing the parameters, i.e., by replacing each $\ddiamondle{r_j}{x}\psi_j$ by $\ddiamond{r_j}\psi_j$, and let $\aut_{\varphi'} = (Q, \pow{P'}, q_0, \delta, F)$ be the equivalent alternating Büchi automaton given by Theorem~\ref{theorem_autconstruction}.
For $j \in \set{1, \ldots, k}$, we denote the set of states of the automaton~$\aut_{r_j}$ checking for a match with $r_j$ by $Q^{r_j}$, which is a subset of $Q$. Furthermore, we assume the $Q^{r_j}$ to be pairwise disjoint. 

\begin{lemma}
	\label{lemma_runcharac}
Let  $\varphi$ and $\aut_{\varphi'}$ as above and let $w \in (\pow{P'})^\omega$. Then, $(w, \alpha) \models \varphi$ if and only if $\aut_{\varphi'}$ has an accepting run~$\rho$ that satisfies the \emph{bounded-match property}: every path~$(q_n, n) \cdots (q_{n+ \ell}, n+\ell)$ in $\rho$ with $q_n, \ldots, q_{n+\ell} \in Q^{r_j}$ satisfies $\ell \le \alpha(x)$ for every $j\in \set{1, \ldots, k}$.
\end{lemma}

To prove this lemma, we first need to strengthen Lemma~\ref{lemma_autrcorrectness} to be able to deal with parameterized formulas in the tests. Fix a regular expression~$r$ with tests~$\theta_1?, \ldots, \theta_k?$, which might contain parameterized operators and let $\aut_r$ the $\epsilon$-NFA with markings obtained from the construction described above Lemma~\ref{lemma_autrcorrectness}. Note that the markings are the formulas~$\theta_1?, \ldots, \theta_k?$. Furthermore, let $w = w_0 w_1 w_2 \cdots \in (\pow{P'})^\omega$, and let $w_0 \cdots w_{n-1}$ be a (possibly empty, if $n=0$) prefix of $w$. The following two statements are equivalent for every $\alpha$:
\begin{enumerate}
	\item $\aut_r$ has an accepting run on $w_0 \cdots w_{n-1}$ with $\epsilon$-paths $\pi_i$ such that $(w_{i}w_{i+1}w_{i+2}\cdots, \alpha)  \models \bigwedge \marking(\pi_i)$ for every $i$ in the range~$0 \le i \le n$ .

\item $(0,n) \in \Rexp(r,w, \alpha)$.

\end{enumerate}

\begin{proof}[Proof of Lemma~\ref{lemma_runcharac}]
The following proof is a strengthening of the proof of Theorem~\ref{theorem_autconstruction}. Again, we proceed by induction over the structure of $\varphi$. 

First, we consider the direction from logic to automata. The induction starts for atomic formulas and
the induction steps for disjunction and conjunction are straightforward. Hence, it remains to consider the temporal operators.

Consider $\ddiamond{r}\psi$. If $(w, \alpha) \models \ddiamond{r}\psi$, then there exists a position~$n$ such that $(w_n w_{n+1} w_{n+2} \cdots, \alpha) \models \psi$ and $(0,n) \in \Rexp(r,w,\alpha)$. Hence, due to the strengthening of Lemma~\ref{lemma_autrcorrectness}, there is an accepting run of $\aut_r$ on $w_0 \cdots w_{n-1}$ such that the tests visited during the run are satisfied with respect to $\alpha$ by the appropriate suffixes of $w$. Thus, applying the induction hypothesis yields accepting runs of the appropriate test automata~$\aut_{\theta_j'}$ on these suffixes which satisfy the bounded-match property. Also, there is an accepting run of $\aut_{\psi'}$ on $w_n w_{n+1} w_{n+2} \cdots$ which satisfies  the bounded-match property, again by induction hypothesis. These runs can be \quot{glued} together to build an accepting run of $\aut_{(\ddiamond{r}\psi)'}$ on $w$ satisfying the bounded-match property.

Now, consider $\ddiamondle{r_j}{x}\psi_j$. If $(w, \alpha) \models\ddiamondle{r_j}{x}\psi_j$, then there is a position~$n \le \alpha(x)$ such that $(w_n w_{n+1} w_{n+2}\cdots,\alpha)  \models \psi_j$ and $(0,n) \in \Rexp(r,w, \alpha)$. Recall that we removed the parameter to obtain $\varphi'$. Thus, we can argue as in the previous case and obtain runs of $\aut_{r_j}$, of the appropriate test automata~$\aut_{\theta_j'}$, and of $\aut_{\psi'}$, all satisfying the induction hypothesis. In particular, the run of $\aut_{r_j}$ has length~$n \le \alpha(x)$ and therefore satisfies  the bounded-match property. Thus, the glued run of $\aut_{(\ddiamondle{r_j}{x}\psi_j)'} = \aut_{\ddiamond{r_j}\psi_j'}$ satisfies the bounded-match property as well.

The case for $\bbox{r}\psi$ is dual to the one for $\ddiamond{r}\psi$, while the cases for the change\-point-bounded operators~$\ddiamondcp{r}\psi$ and $\bboxcp{r}\psi$ are analogous, using the fact that $\aut_{cp}$ accepts words which have at most one changepoint. 

Now, we consider the other direction, where the induction starts for atomic formulas and the induction steps for disjunction and conjunction are again straightforward. 

We continue with formulas of the form~$\ddiamond{r}\psi$. Let $\rho$ be an accepting run of $\aut_{(\ddiamond{r}\psi)'}$ on $w$. Let $n \ge 0$ be the last level of $\rho$ that contains a state from $Q^r$. Such a level has to exist since states in $Q^r$ are not accepting and they have no incoming edges from states of the automata~$\aut_{\psi'}$ and $\aut_{\theta_j'}$ (the $\theta_j'$ are the tests in $r$), but the initial state of $\aut_{(\ddiamond{r}\psi)'}$ is in $Q^r$. Furthermore, $\aut_{(\ddiamond{r}\psi)'}$ is non-deterministic when restricted to states in $Q^r \setminus C^r$. Hence, we can extract an accepting run of $\aut_r$ from $\rho$ on $w_0 \cdots w_{n-1}$ that additionally satisfies the requirements formulated in the strengthening of Lemma~\ref{lemma_autrcorrectness}, due to the transitions into the test automata and an application of the induction hypothesis. Hence, we have $(0,n) \in \Rexp(r,w,\alpha) $. Also, from the remainder of $\rho$ (levels greater or equal to $n$) we can extract an accepting run of $\aut_{\psi'}$ on $w_n w_{n+1} w_{n+2} \cdots$ satisfying the bounded-match property. Hence, $(w_n w_{n+1} w_{n+2}\cdots , \alpha)  \models \psi$ by induction hypothesis. So, we conclude $(w, \alpha) \models \ddiamond{r}\psi$.

In the case of $\ddiamondle{r_j}{x}\psi_j$, the reasoning is similar: we have removed the parameter to obtain $\varphi'$. Thus, we end up in an analogous situation as in the previous case, but the level~$n$ satisfies $n \le \alpha(x)$ due to the bounded-match property. This implies $(w, \alpha) \models \ddiamondle{r_j}{x}\psi_j$.

Again, the case for $\bbox{r}\psi$ is dual to the one for $\ddiamond{r}\psi$ and the cases for the changepoint-bounded operators~$\ddiamondcp{r}\psi$ and $\bboxcp{r}\psi$ rely on the fact that $\aut_{cp}$ only accepts words which have at most one changepoint.
\end{proof}

Next, we show that words having runs as described in Lemma~\ref{lemma_runcharac} can be recognized by a non-deterministic Büchi automaton: the following lemma concludes the proof of Item~\ref{thm_smallautomata_nondet} of Theorem~\ref{thm_smallautomata}. To this end, we extend the classical Breakpoint construction~\cite{MiyanoH84} by counters that check the bounded-match property: the original construction yields an automaton that guesses an accepting run of a given alternating Büchi automaton level by level, which are represented as the set of states they contain. We employ the counters~$\gamma$ to keep track of the length of paths in $Q^{r_j}$ in the guessed run. If the bound $\alpha(x)$ is exceeded, then the guessed run is discarded.   

\begin{lemma}
	\label{lemma_runreadingaut}
There exists a non-deterministic Büchi automaton of size~$(3 \cdot(\alpha(x)+1))^{\mathcal{O}(\size{\varphi})}$ that accepts $w \in (\pow{P'})^\omega$ if and only if $\aut_{\varphi'}$ has an accepting run on $w$ satisfying the bounded-match property.
\end{lemma}

\begin{proof}
Let $\aut_{\varphi'} = (Q, \pow{P'}, q_0, \delta, F)$ be as above
and recall that $Q^{r_j} \subseteq Q$ for $j \in \set{1, \ldots, k}$ is the set of states which has to be left after at most $\alpha(x)$ steps in order to satisfy the bounded-match property. Define the Büchi automaton~$\aut' = (Q', \pow{P'}, q_0', \delta', F')$ with
\begin{itemize}

	\item $Q' = \set{(T, O, \gamma) \mid Q \supseteq T \supseteq O \text{ and } \gamma \in \set{0, 1, \ldots, \alpha(x)}^{T \cap \bigcup_{j =1}^k Q^{r_j}}}$,

	\item $q_0' = (\set{q_0}, \emptyset, \gamma)$, where $\gamma(q_0) = \alpha(x)$ if $q_0 \in \bigcup_{j=1}^k Q^{r_j}$,

	\item $F' = \set{ (T, \emptyset, \gamma) \mid (T, \emptyset, \gamma) \in Q' }$, and 

	\item $ \delta'((T, O, \gamma),A) $ is equal to 

\end{itemize}
 	\begin{align*}
 		 \{ (T',T' \setminus F,\update{\gamma}{G}) \mid & \text{ exists graph $G = (T \cup T', E)$ with $E \subseteq T \times T'$}\\
 		& \text{s.t.\ $\suc{G}(q) \models \delta(q, A)$ for every $q \in T$}\} \cap Q'
 	\end{align*}
if $O = \emptyset$, and equal to 
 	\begin{align*}
 		 \{ (T',O' \setminus F,\update{\gamma}{G}) \mid& \text{ $O' \subseteq T'$ and there} \\
 		 & \text{exists graph $G = (T \cup T', E)$ with $E \subseteq T \times T'$}\\
 		& \text{s.t.\ $\suc{G}(q) \models \delta(q, A)$ for every $q \in T$, and}\\
 		& \text{$\suc{G\upharpoonright(O \cup O')}(q) \models \delta(q, A)$ for every $q \in O$}\} \cap Q'
 	\end{align*}
 	if $O \neq \emptyset$.
Here, $\suc{G}(q)$ denotes the set of successors of $q$ in $G$, $G \upharpoonright (O \cup O')$ is the restriction of $G$ to $O \cup O'$, and $\update{\gamma}{G}$ is defined via
\[\update{\gamma}{G}(q') = \min \set{\alpha(x), \gamma(q)-1 \mid (q,q') \in E \text{ and $q,q' \in Q^{r_j}$ for some $j$}}.\]
Note that we might have $\update{\gamma}{G}(q') < 0$, which implies that $\update{\gamma}{G}$ is not the third component of a state of $\aut'$ and explains the intersection with $Q'$ in the definition of $\delta'$. Thus, the counter~$\gamma$ prevents the simulation of runs of $\aut_{\varphi'}$ that violate the bounded-match property by blocking transitions. 

Intuitively, the graphs used to define the transition relation~$\delta'$ are building blocks for runs of $\aut_{\varphi'}$ that contain two levels of a run as well as the edges between them that witness the satisfaction of the transition relation~$\delta$. As already explained, the counter~$\gamma$ ensures that every path through some $Q^{r_j}$ is of length $\alpha(x)$ or less.  

In the first two components of $\aut'$, we implement the Breakpoint construction while we use the third component to implement a counter that checks the bounded-match property. The correctness of this construction follows directly from the correctness of the Breakpoint construction~\cite{MiyanoH84}. 
\end{proof}

In case $\varphi$ has parameterized box-operators, e.g., with variable~$y$, Lemma~\ref{lemma_runcharac} reads as follows: 

\begin{quote}
$(w, \alpha) \models \varphi$ if and only if $\aut_{\varphi'}$ has an accepting run~$\rho$ where every path of the form~$(q_n, n) \cdots (q_{n+ \ell}, n+\ell)$ in $\rho$ with $q_n, \ldots, q_{n+\ell} \in Q^{r_j}$ for some $j \in \set{1, \ldots, k}$ satisfying $\ell > \alpha(x)$ may end in a terminal vertex~$(q_{n+ \ell}, n+\ell)$.
\end{quote}
As before, adapting the Breakpoint construction by adding a counter mapping states in ${T \cap \bigcup_{j =1}^k Q^{r_j}}$ to $\set{0, 1, \ldots, \alpha(x)}$ yields a non-deterministic Büchi automaton that accepts exactly those words having a run that satisfies the bounded-match property for formulas with parameterized box-operators.

\section{Conclusion}
\label{sec_conc}
We introduced Parametric Linear Dynamic Logic, which extends Linear Dynamic Logic by temporal operators equipped with parameters that bound their scope, similarly to Parametric Linear Temporal Logic, which extends Linear Temporal Logic by parameterized temporal operators. Here, the model checking problem asks for a valuation of the parameters such that the formula is satisfied with respect to this valuation on every path of the transition system. Realizability  is defined in the same spirit. 

We showed $\pldl$ model checking and $\pldl$ assume-guarantee model checking to be  $\pspace$-complete and  $\pldl$ realizability to be $\twoexp$-complete, just as for $\ltl$. Thus, in a sense, $\pldl$ is not harder than $\ltl$. Finally, we were able to give tight exponential respectively doubly-exponential bounds on the optimal valuations for model checking and realizability.

With respect to the computation of optimal valuations, we have shown this to be possible in polynomial space for model checking and in triply-exponential time for realizability, which is similar to the situation for $\pltl$~\cite{AlurEtessamiLaTorrePeled01,Zimmermann13}. Note that it is an open question whether optimal valuations for $\pltl$ realizability can be determined in doubly-exponential time. Recently, a step towards this goal was made by giving an $\frac{1}{2}$-approximation algorithm with doubly-exponential running time~\cite{TentrupWeinertZimmermann15}.

\bibliographystyle{splncs03}
\bibliography{biblio}

\newcommand{\noopsort}[1]{}
\begin{thebibliography}{10}
\providecommand{\url}[1]{\texttt{#1}}
\providecommand{\urlprefix}{URL }

\bibitem{AlurEtessamiLaTorrePeled01}
Alur, R., Etessami, K., Torre, S.L., Peled, D.: Parametric temporal logic for
  ``model measuring''. ACM Trans. Comput. Log.  2(3),  388--407 (2001)

\bibitem{Forspec02}
Armoni, R., Fix, L., Flaisher, A., Gerth, R., Ginsburg, B., Kanza, T., Landver,
  A., Mador-Haim, S., Singerman, E., Tiemeyer, A., Vardi, M.Y., Zbar, Y.: The
  {ForSpec} temporal logic: A new temporal property-specification language. In:
  Katoen, J.P., Stevens, P. (eds.) TACAS 2002. LNCS, vol. 2280, pp. 296--311.
  Springer (2002)

\bibitem{BuechiLandweber69}
B\"uchi, J.R., Landweber, L.H.: Solving sequential conditions by finite-state
  strategies. Trans. Amer. Math. Soc.  138,  pp. 295--311 (1969)

\bibitem{GiacomoVardi13}
{De Giacomo}, G., Vardi, M.Y.: Linear temporal logic and linear dynamic logic
  on finite traces. In: Rossi, F. (ed.) IJCAI. IJCAI/AAAI (2013)

\bibitem{EisnerFismanPSL}
Eisner, C., Fisman, D.: A Practical Introduction to PSL. Integrated Circuits
  and Systems, Springer (2006)

\bibitem{EmersonJutla91}
Emerson, E.A., Jutla, C.S.: Tree automata, mu-calculus and determinacy
  (extended abstract). In: FOCS 1991. pp. 368--377. IEEE (1991)

\bibitem{FaymonvilleZimmermann14}
Faymonville, P., Zimmermann, M.: Parametric linear dynamic logic. In: Peron,
  A., Piazza, C. (eds.) GandALF 2014. {EPTCS}, vol. 161, pp. 60--73 (2014)

\bibitem{FiliotJinRaskin2011}
Filiot, E., Jin, N., Raskin, J.F.: Antichains and compositional algorithms for
  {LTL} synthesis. Formal Methods in System Design  39(3),  261--296 (2011)

\bibitem{FinkbeinerSchewe2013}
Finkbeiner, B., Schewe, S.: Bounded synthesis. STTT  15(5-6),  519--539 (2013)

\bibitem{FischerLadner1979}
Fischer, M.J., Ladner, R.E.: Propositional dynamic logic of regular programs.
  Journal of Computer and System Sciences  18(2),  194 -- 211 (1979)

\bibitem{GastinO01}
Gastin, P., Oddoux, D.: Fast {LTL} to {Büchi} automata translation. In: Berry,
  G., Comon, H., Finkel, A. (eds.) CAV 2001. LNCS, vol. 2102, pp. 53--65.
  Springer (2001)

\bibitem{GiampaoloTorreNapoli15}
Giampaolo, B.D., {La Torre}, S., Napoli, M.: Parametric metric interval
  temporal logic. Theor. Comput. Sci.  564,  131--148 (2015)

\bibitem{Kamp68}
Kamp, H.W.: Tense Logic and the Theory of Linear Order. Ph.D. thesis, Computer
  Science Department, University of California at Los~Angeles, USA (1968)

\bibitem{KupfermanPitermanVardi09}
Kupferman, O., Piterman, N., Vardi, M.Y.: From liveness to promptness. Formal
  Methods in System Design  34(2),  83--103 (2009)

\bibitem{LeuckerSanchez07}
Leucker, M., S\'{a}nchez, C.: Regular linear temporal logic. In: Jones, C.,
  Liu, Z., Woodcock, J. (eds.) ICTAC 2007. LNCS, vol. 4711, pp. 291--305.
  Springer-Verlag, Macau, China (September 2007)

\bibitem{MiyanoH84}
Miyano, S., Hayashi, T.: Alternating finite automata on $\omega$-words. Theor.
  Comput. Sci.  32,  321--330 (1984)

\bibitem{Mostowski91}
Mostowski, A.: Games with forbidden positions. Tech. Rep.~78, University of
  Gda\'nsk (1991)

\bibitem{MullerSS92}
Muller, D.E., Saoudi, A., Schupp, P.E.: Alternating automata, the weak monadic
  theory of trees and its complexity. Theor. Comput. Sci.  97(2),  233--244
  (1992)

\bibitem{Pnueli77}
Pnueli, A.: The temporal logic of programs. In: FOCS 1977. pp. 46--57. IEEE
  (Oct 1977)

\bibitem{Pnueli85}
Pnueli, A.: In transition from global to modular temporal reasoning about
  programs. In: Apt, K.R. (ed.) Logics and Models of Concurrent Systems, NATO
  ASI Series, vol.~13, pp. 123--144. Springer (1985)

\bibitem{PnueliRosner89a}
Pnueli, A., Rosner, R.: On the synthesis of an asynchronous reactive module.
  In: Ausiello, G., Dezani-Ciancaglini, M., Rocca, S.R.D. (eds.) ICALP 1989.
  LNCS, vol. 372, pp. 652--671. Springer (1989)

\bibitem{Rohde:1997:AAT:925874}
Rohde, G.S.: Alternating Automata and the Temporal Logic of Ordinals. Ph.D.
  thesis, University of Illinois at Urbana-Champaign, Champaign, IL, USA (1997)

\bibitem{Schewe09}
Schewe, S.: Tighter bounds for the determinisation of {B}üchi automata. In:
  de~Alfaro, L. (ed.) FOSSACS 2009. LNCS, vol. 5504, pp. 167--181. Springer
  (2009)

\bibitem{AlthoffThomasWallmeier2006}
{Schulte Althoff}, C., Thomas, W., Wallmeier, N.: Observations on
  determinization of {B}üchi automata. Theor. Comput. Sci.  363(2),  224 --
  233 (2006)

\bibitem{SistlaClarke85}
Sistla, A.P., Clarke, E.M.: The complexity of propositional linear temporal
  logics. J. ACM  32(3),  733--749 (1985)

\bibitem{TentrupWeinertZimmermann15}
Tentrup, L., Weinert, A., Zimmermann, M.: Approximating optimal bounds in
  {P}rompt-{LTL} realizability in doubly-exponential time (2015),
  arXiv:1511.09450

\bibitem{Thompson68}
Thompson, K.: Programming techniques: Regular expression search algorithm.
  Commun. ACM  11(6),  419--422 (Jun 1968)

\bibitem{Vardi11}
Vardi, M.Y.: The rise and fall of {LTL}. In: D'Agostino, G., Torre, S.L. (eds.)
  GandALF 2011. {EPTCS}, vol.~54 (2011)

\bibitem{VardiWolper94}
Vardi, M.Y., Wolper, P.: Reasoning about infinite computations. Inf. Comput.
  115(1),  1--37 (1994)

\bibitem{Wolper1983}
Wolper, P.: Temporal logic can be more expressive. Information and Control
  56(1–2),  72 -- 99 (1983)

\bibitem{Zimmermann13}
Zimmermann, M.: Optimal bounds in parametric {LTL} games. Theor. Comput. Sci.
  493,  30--45 (2013)

\end{thebibliography}

\end{document}